\documentclass[12pt]{amsart}
\usepackage{}

\usepackage{amsmath}
\usepackage{amsfonts}
\usepackage{amssymb}
\usepackage[all]{xy}           

\usepackage{bbding}
\usepackage{txfonts}
\usepackage{amscd}

\usepackage[shortlabels]{enumitem}
\usepackage{ifpdf}
\ifpdf
  \usepackage[colorlinks,final,backref=page,hyperindex]{hyperref}
\else
  \usepackage[colorlinks,final,backref=page,hyperindex,hypertex]{hyperref}
\fi
\usepackage{tikz}
\usepackage[active]{srcltx}

\makeatletter

\newtheorem{thm}{Theorem}[section]
\newtheorem{lem}[thm]{Lemma}
\newtheorem{cor}[thm]{Corollary}
\newtheorem{pro}[thm]{Proposition}
\newtheorem{ex}[thm]{Example}
\newtheorem{rmk}[thm]{Remark}
\newtheorem{defi}[thm]{Definition}
\newtheorem{propdef}[thm]{Proposition-definition}

\newcommand {\emptycomment}[1]{}
\newcommand {\yh}[1]{{\marginpar{*}\scriptsize\textcolor{purple}{yh: #1}}}

\newcommand{\nc}{\newcommand}
\newcommand{\delete}[1]{}
\newcommand{\U}{\mathrm{U}}

\nc{\CV}{\mathbf{C}}

\newcommand{\CE}{\mathsf{CE}}

\newcommand{\LR}{$\mathsf{Lie}\mathsf{Rep}$~}

\newcommand{\ET}{\mathsf{ET}}

\newcommand{\lon }{\,\rightarrow\,}
\newcommand{\be }{\begin{equation}}
\newcommand{\ee }{\end{equation}}

\newcommand{\g}{\mathfrak g}
\newcommand{\h}{\mathfrak h}
\newcommand{\la}{\mathfrak G} 

\newcommand{\Real}{\mathbb R}

\newcommand{\huaL}{\mathcal{L}}

\newcommand{\huaP}{\mathcal{P}}
\newcommand{\huaD}{\mathcal{D}}

\newcommand{\huaH}{\mathcal{H}}

\newcommand{\huaT}{\mathcal{T}}

\newcommand{\frki}{\mathfrak i}

\newcommand{\frkl}{\mathfrak l}

\newcommand{\frkp}{\mathfrak p}

\newcommand{\frks}{\mathfrak s}

\newcommand{\frkC}{\mathfrak C}

\newcommand{\frkT}{\mathfrak T}

\newcommand{\half}{\frac{1}{2}}

\newcommand{\Courant}[1]{\left\llbracket  #1\right\rrbracket }

\newcommand{\Id}{{\rm{Id}}}

\newcommand{\br}[1]{   [ \cdot,    \cdot  ]   }

\newcommand{\dM}{\mathrm{d}}

\newcommand{\Hom}{\mathrm{Hom}}

\newcommand{\Sym}{\mathsf{S}}
\newcommand{\Ten}{\mathsf{T}}
\newcommand{\KS}{\mathsf{KS}}

\newcommand{\Leib}{\mathsf{Leibniz}}
\newcommand{\Lie}{\mathsf{Lie}}
\newcommand{\Alg}{\mathsf{Alg}}

\newcommand{\ann}{\mathsf{ann}}

\newcommand{\Der}{\mathrm{Der}}

\newcommand{\gl}{\mathfrak {gl}}

\newcommand{\ol}{\mathfrak {ol}}
\newcommand{\so}{\mathfrak {so}}

\nc{\oprn}{\theta}
\newcommand{\B}{\mathsf{B}}
\newcommand{\NR}{\mathsf{NR}}
\newcommand{\co}{\mathsf{cosh}}

\newcommand{\End}{\mathrm{End}}
\newcommand{\ad}{\mathrm{ad}}

\newcommand{\Img}{\mathrm{Im}}

\newcommand{\K}{\mathbb{K}}

\begin{document}

\title[The controlling $L_\infty$-algebra, cohomology and homotopy]{The controlling $L_\infty$-algebra, cohomology and homotopy of embedding tensors and Lie-Leibniz triples
}
\author{Yunhe Sheng}
\address{Department of Mathematics, Jilin University, Changchun 130012, Jilin, China}
\email{shengyh@jlu.edu.cn}

\author{Rong Tang}
\address{Department of Mathematics, Jilin University, Changchun 130012, Jilin, China}
\email{tangrong@jlu.edu.cn}

\author{Chenchang Zhu}
\address{ Mathematics Institute, Georg-August-University G\"ottingen, Bunsenstrasse 3-5, 37073, G\"ottingen, Germany}
\email{czhu@gwdg.de }


\begin{abstract}
In this paper, we first construct the controlling algebras of
embedding tensors and Lie-Leibniz triples, which turn out to be a
graded Lie algebra and an $L_\infty$-algebra respectively. Then we
introduce representations and cohomologies of embedding tensors and
Lie-Leibniz triples, and show that there is a long exact sequence
connecting various cohomologies. As applications, we classify
infinitesimal deformations and central extensions using the second
cohomology groups. Finally, we introduce the notion of a homotopy
embedding tensor which will induce a Leibniz$_\infty$-algebra. We
realize Kotov and Strobl's construction of an $L_\infty$-algebra
from an embedding tensor, as a functor from the category of homotopy
embedding tensors to that of Leibniz$_\infty$-algebras, and a functor further to that of $L_\infty$-algebras.

\end{abstract}

\subjclass[2010]{17B40, 17B56,  70S15}

\keywords{embedding tensor, averaging operator,  Lie-Leibniz triple, cohomology, deformation, extension, homotopy, Leibniz$_\infty$-algebra}

\maketitle

\tableofcontents

\allowdisplaybreaks


\section{Introduction}
An embedding tensor on a Lie algebra representation $(\g,V)$ is a
linear map $T:V\lon\g$ satisfying a quadratic equivariancy constraint (see Definition \ref{defi:O}). Leibniz algebras,  embedding tensors and their associated tensor hierarchies provide a
nice and efficient way to construct supergravity theories
and further to higher gauge theories (see e.g. \cite{BH,BH-1,Hoh, KS, Str19} and
references therein
for a rich physics literature on this subject, and see \cite{Lavan,LavauP}
 for a math-friendly introduction on this
subject).  Recently, this topic has attracted much attention of
the mathematical physics world. First of all, a sharp and beautiful observation by Kotov and Strobl in a
recent article \cite{KS} demonstrates for us a possible mathematical nature
behind the various calculations from embedding tensors to their
associated tensor hierarchies in the physics literature. An embedding
tensor gives rise to a Leibniz algebra, which further gives rise to an
$L_\infty$-algebra, and this corresponds to tensor hierarchies in physics
literature. We see later that both procedures are
functorial (the first one in Section \ref{sec:L}, and the second one in Section
\ref{sec:homotopy}), moreover the functoriality can be extended with
{\em homotopy} added in for all objects. In
fact, the second procedure is a composition of several very classic results
\cite{LV, Milnor}. This therefore guarantees us, from a category viewpoint, that
the process from embedding tensors to tensor hierarchies, and its
corresponding transition from supergravity theories to higher gauge
theories, is natural.  Then, almost at the same time, appeared several other approaches to
encode this process to tensor hierarchies, involving also Leibniz
field theory: \cite{Str16, SW} are from
the point of view of enhanced Leibniz structures; \cite{LavauS} builds an
$L_\infty$-algebra extension from a Leibniz algebra, which is
apparently different from the second functor described above.  The
functoriality was shown in both procedures.

In our setting, we further conjecture that the above two functors are functorial
in an $\infty$-category sense. We will explore this direction in a future work
\cite{STZ:2}. Notice that the homotopy nature of $L_\infty$-algebras
suggests homotopy also in the category hosting these objects.

In this article, we provide a rich math tool box for embedding tensors
and Lie-Leibniz triples, which seem not yet existing in the mathematical
literature.  Indeed, as a sort of algebra  (or operad), embedding tensors involve not only binary but also unary operations. We develop the theory of  controlling algebras, thus further the theory of cohomology and homotopy for embedding tensors and Lie-Leibniz triples.

To establish a good cohomology theory for an object, besides the
standard homological algebraic method of projective resolutions, there is
also another shorter way through its ``controlling
algebraic object''. Let us explain this idea in the case of a Lie algebra
$\g$. We start with a vector space $\g$, then the graded vector space
$\oplus_{k=0}^{+\infty}\Hom(\wedge^k \g, \g)$ 
equipped with the Nijenhuis-Richardson bracket $[\cdot,\cdot]_{\NR}$ becomes a
graded Lie algebra (g.l.a.), or a differential graded Lie algebra
(d.g.l.a.) with 0 differential \cite{NR}. Then a Lie algebra structure on $\g$
corresponds {\em exactly} to a Maurer-Cartan element $\pi \in
\Hom(\wedge^2 \g, \g)$. We call this g.l.a. $
(\oplus_{k=0}^{+\infty}\Hom(\wedge^k \g, \g), [\cdot,\cdot]_{\NR})$ the {\em
  controlling algebra of Lie algebra structures} on $\g$. Furthermore,
since $[\pi, \pi]_{\NR}=0$,  $d_{\pi}:=[\pi, \cdot]_{\NR}$ satisfies
$d^2_{\pi}=0$, thus $d_{\pi}$ is a differential.  The controlling
g.l.a. $\oplus_{k=0}^{+\infty}\Hom(\wedge^k \g, \g)$ together with
$d_{\pi}$ becomes exactly the Chevalley-Eilenberg complex to
calculate the cohomology of $\g$ with coefficients in its adjoint
representation. This is a general phenomenon and works not only for
Lie algebras, but also for associative algebras, Leibniz algebras,
$n$-Lie algebras, and pre-Lie algebras.  See the review \cite{review,GLST1}
for more details. Thus, we use this principal as a guide   to
develop   cohomology theories for embedding tensors (Section \ref{sec:A}) and
Lie-Leibniz triples (Section \ref{sec:control}). Here a Lie-Leibniz triple \cite{LavauP} consists
of a Lie algebra representation $(\g, V)$,  and an embedding tensor
$T: V\to \g$. The subtle difference between these two concepts shows
up, e.g. in deformation theory. To deform an
embedding tensor, we fix   $(\g, V)$  and deform only the operator
$T$, but to deform a Lie-Leibniz triple, we are allowed to deform also
the Lie algebra and its representation $(\g, V)$ simultaneously. It turns out that the controlling algebraic structure for embedding tensors is a g.l.a. and that for Lie-Leibniz triples is an $L_\infty$-algebra. Thus indeed the theory of Lie-Leibniz triples is more involved.

We give some immediate applications (also as verifications) of this cohomology theory of Lie-Leibniz triples in Section \ref{sec:cohomology}. It does behave as it should: Given a Lie-Leibniz triple,
\begin{enumerate}
    \item its second cohomology classes in $H^2$ with coefficients in the adjoint representation correspond exactly to the equivalence classes of its infinitesimal deformations;
    \item its second cohomology classes in $H^2$ with coefficients in the trivial representation correspond exactly to  the equivalence classes of central extensions.
\end{enumerate}
Here, we actually need a bit of additional luck in the second
application: We need to develop a cohomology theory for Lie-Leibniz
triples with arbitrary coefficients, but not just the one from adjoint
representation. For this, we find that there is a natural
 Lie-Leibniz triple structure on the endomorphisms of a 2-term complex of vector
spaces. This natural structure comes from the strict Lie 2-algebra
structure on them \cite{shengzhu2}, and a strict Lie 2-algebra is a
special Lie-Leibniz triple (see Example \ref{example-1} and Example
\ref{example-2}).

Finally, in Section \ref{sec:homotopy}, we study how embedding tensors cooperate with
homotopy.  That is, what a homotopy embedding tensor should be, and
how Kotov-Strobl's functor $\KS$ behaves with respect to homotopy. Will $\KS$
still produces an $L_\infty$-algebra or something involving even more
homotopy? This will test how stable the concept of embedding tensors
and the procedure to topological hierarchies are. We still use the tool of controlling algebras to develop  the homotopy theory. A
standard approach \cite{LV, Vallette-1} to give  a homotopy $\huaP$-algebraic structure is
to construct  a minimal model $\huaP_\infty$ of the  operad
$\huaP$. Along this approach, $L_\infty$-algebras and
$A_\infty$-algebras are well developed. Moreover,
Leibniz$_\infty$-algebras are defined as the algebras over the
minimal model \cite{ammardefiLeibnizalgebra,livernet} of the $\huaL
eibniz$ operad. However, apart from this approach, we can also use Maurer-Cartan
elements of the aforementioned controlling algebra on a graded vector
space to define a homotopy algebraic structure. For example, to define
an $L_\infty$-algebra, one can start with a graded vector space
$\g^\bullet$ and define an $L_\infty$-algebra to be a Maurer-Cartan
element of the g.l.a. $ (\Hom(\Sym( \g^\bullet), \g^\bullet),
[\cdot,\cdot]_{\NR})$. Using this method, we define a homotopy embedding
tensor to be a Maurer-Cartan element of a graded
version of the controlling
algebra that we develop in Section \ref{sec:A}. Then we show that a homotopy embedding tensor gives rise to a Leibniz$_\infty$-algebra, and a
Leibniz$_\infty$-algebra further gives rise to an
$L_\infty$-algebra. We further prove that these two processes are
functorial. Thus the functor $\KS$ extends to a homotopic version.

We want to emphasis that embedding tensors and Lie-Leibniz triples
have been already known in mathematics literature under the name of
averaging operators and averaging algebras respectively for a long time. In the last century, many studies on averaging operators were done for various special algebras, such as function spaces, Banach algebras, and the topics and methods were largely
analytic \cite{Barnett,Brainerd,Huijsmans,Rota}. See
the well-written introduction in \cite{PG} for more details. More
recently,  people have begun to study averaging operators in double
algebras, classical Yang-Baxter equation, conformal algebras, and the
procedure of replication in the operad theory \cite{Aguiar,
  Goncharov,Kolesnikov,Pei-Bai-Guo-Ni}. It is not yet clear to us how these aspects of embedding tensors and Lie-Leibniz triples are related. But we wish our article makes some first steps to understand these concepts more deeply.

\noindent
{\bf Acknowledgements. } Y. Sheng is  supported by National Science Foundation of China
(11922110). R. Tang is supported by National Science Foundation of China
(12001228) and China Postdoctoral Science
Foundation (2020M670833). C. Zhu is funded by Deutsche Forschungsgemeinschaft (ZH 274/1-1, ZH
243-3-1, RTG 2491). We thank warmly Florian Naef, Dmitry
Roytenberg, Jim Stasheff, and Bruno Vallette for very helpful discussions and
suggestions. We also thank ESI, Vienna, for their invitation to present a
preliminary version of this work during the Programme on Higher Structures and Field Theory.

\section{Embedding tensors, omni-Lie algebras and Leibniz algebras}\label{sec:L}

In this section, first  we establish relations between embedding tensors, omni-Lie algebras and Leibniz algebras. Then we give some interesting examples of embedding tensors.

\begin{defi}
  A {\em \LR pair} consists of a Lie algebra  $(\g,[\cdot,\cdot]_\g)$  and a representation $\rho:\g\longrightarrow\gl(V)$   of $\g$ on a vector space $V$.
\end{defi}

We denote a \LR pair by $((\g,[\cdot,\cdot]_\g),(V;\rho))$, or simply by $(\g,V)$ if there is no confusion.


\begin{defi} \label{defi:O}
\begin{enumerate}
\item[\rm(i)] An {\bf embedding tensor}  on a \LR pair $((\g,[\cdot,\cdot]_\g),(V;\rho))$ is a linear map $T:V\longrightarrow\g$ satisfying the following quadratic constraint:
 \begin{equation}
   [Tu,Tv]_\g=T\big(\rho(Tu)(v)\big),\quad\forall u,v\in V.
 \end{equation}
\item[\rm(ii)]

 A {\bf Lie-Leibniz triple} is a triple $(\g,V,T)$, where  $(\g,V)$ is a \LR pair and  $T:V\longrightarrow\g$ is an embedding tensor on the \LR pair $(\g,V)$.
\end{enumerate}
\end{defi}

      \begin{defi}\label{defi:homoRRB}
   Let $((\g,[\cdot,\cdot]_\g),(V;\rho),T)$ and $((\g',\{\cdot,\cdot\}_{\g'}),(V',\rho'),T')$ be two Lie-Leibniz triples. A {\bf homomorphism} from  $((\g',\{\cdot,\cdot\}_{\g'}),(V',\rho'),T')$ to $((\g,[\cdot,\cdot]_\g),(V;\rho),T)$ consists of
     a Lie algebra homomorphism  $\phi:\g'\longrightarrow\g$ and a linear map $\varphi:V'\longrightarrow V$ such that
         \begin{eqnarray}
          T\circ \varphi&=&\phi\circ T',\label{defi:isocon1}\\
                \varphi\rho'(x)(u)&=&\rho(\phi(x))(\varphi(u)),\quad\forall x\in\g', u\in V'.\label{defi:isocon2}
      \end{eqnarray}

      In particular, if $\phi$ and $\varphi$ are  invertible,  then $(\phi,\varphi)$ is called an  {\bf isomorphism}.

\end{defi}

The algebraic structure underlying an  embedding tensor is a {\bf Leibniz algebra}, which is a vector space $\la$ together with a bilinear operation $[\cdot,\cdot]_\la:\la\otimes\la\lon\la$ such that
\begin{eqnarray*}
\label{Leibniz}[x,[y,z]_\la]_\la=[[x,y]_\la,z]_\la+[y,[x,z]_\la]_\la,\quad\forall x,y,z\in\la.
\end{eqnarray*}
\begin{pro}{\rm (\cite{Aguiar})}\label{average-Leibniz}
Let $T:V\lon\g$ be an embedding tensor on a \LR pair $((\g,[\cdot,\cdot]_\g),(V;\rho))$. Then there exists a Leibniz algebra structure $[\cdot,\cdot]_T$ on $V$ given by
\begin{eqnarray}
[u,v]_T:=\rho(Tu)v,\,\,\,\,\forall u,v\in V.
\end{eqnarray}

\end{pro}
\begin{rmk}
This association  gives rise to a functor $F: \ET\to \Leib\mbox{-}\Alg$ from the
category of embedding tensors to that of Leibniz algebras. 
The direction from Leibniz algebras to embedding tensors is less well
behaved. It is easy to check that the association in \cite{KS}
gives rise to a functor $G: \Leib\mbox{-}\Alg \to
\ET$, and $F\circ G=\Id$ as also noticed therein.  But $G\circ F
\neq \Id$, and these two functors do not
differ even by a natural transformation. There is
another natural association given in \cite{LavauP}, namely for a Leibniz algebra $(\la,[\cdot,\cdot]_\la)$,    the left multiplication $L:\la\lon\gl(\la)$ given by
\begin{equation}\label{eq:lm}
L_xy=[x,y]_\la,\quad \forall x,y\in \la,
\end{equation} is an embedding tensor on the Lie algebra $\gl(\la)$ with
respect to the natural representation on the vector space
$\la$. Even though with this method, it is more likely to create a
natural transformation, it does not give rise to even a functor $\Leib\mbox{-}\Alg\to
\ET$ in the first place.
\end{rmk}


In the sequel, we give an alternative explanation of Proposition \ref{average-Leibniz} using  integrable subspaces of omni-Lie algebras. For this purpose, we give an interesting example of embedding tensors.

\begin{ex}\label{example-4}
{\rm
  Let $V$ be a vector space. Then the general linear Lie algebra $\gl(V)$ represents on the direct sum $\gl(V)\oplus V$ naturally via:
  $$
  \rho(A)(B+v)=[A,B]+Av,\quad \forall A,B\in\gl(V), v\in V.
  $$
  Let $P:\gl(V)\oplus V\lon \gl(V)$ be the projection to the first summand. Then we have
  $$
  P\Big(\rho(P(A+u))(B+v)\Big)=P(\rho(A)(B+v))=[A,B]=[P(A+u),P(B+v)],
  $$
for all $A,B\in\gl(V), u,v\in V.$ Thus, $P$ is an embedding tensor on $\gl(V) $ with respect to the representation $(\gl(V)\oplus V;\rho)$.
   }
\end{ex}

 By Proposition \ref{average-Leibniz}, there is an induced  Leibniz algebra structure on $\gl(V)\oplus V$ given by
  \begin{equation}\label{eq:omniLiebra}
 [ A+u,B+v]_\ol=\rho(P(A+u))(B+v)=[A,B]+Av.
\end{equation}

 The above bracket $[\cdot,\cdot]_\ol$ is exactly the omni-Lie bracket  on $\gl(V)\oplus V$ introduced by Weinstein in \cite{Alan}. Recall that an {\bf omni-Lie algebra} is a triple $(\gl(V)\oplus V,[\cdot,\cdot]_\ol,(\cdot,\cdot)_+)$, where the omni-Lie bracket $[\cdot,\cdot]_\ol$ is given by \eqref{eq:omniLiebra}, and $(\cdot,\cdot)_+$ is a symmetric $V$-valued pairing given by
   \begin{equation}\label{eq:omniLiepair}
 ( A+u,B+v)_+= Av+Bu.
\end{equation}

\begin{defi}
  A subspace $H\subset \gl(V)\oplus V$ is said to be {\bf integrable} if $[H,H]_\ol\subset H$. 
\end{defi}

Now we are ready to characterize embedding tensors using integrable subspaces of the omni-Lie algebra.

\begin{thm}
  Let $T:V\lon \gl(V)$ be a linear map. Then $T$ is an embedding tensor on the general linear Lie algebra $\gl(V)$ with respect to the natural representation on $V$ if and only if the graph of $T$, denoted by $G_T$, is an integrable subspace of the omni-Lie algebra $(\gl(V)\oplus V,[\cdot,\cdot]_\ol,(\cdot,\cdot)_+)$.
\end{thm}

\begin{proof}
  For all $u,v\in V$, we have
  \begin{eqnarray*}
    [Tu+u,Tv+v]_\ol=[Tu, Tv]+(Tu)v,
  \end{eqnarray*}
  which implies that the graph of $T$ is integrable if and only if $[Tu,Tv] =T((Tu)v)$, i.e. $T$ is an embedding tensor.
\end{proof}

\begin{rmk}
  Since the omni-Lie bracket $[\cdot,\cdot]_\ol$ is a Leibniz algebra
  structure, it follows that an integrable subspace is also a Leibniz
  algebra. Thus, if $T:V\lon\gl(V)$ is an embedding tensor, then $G_T$
  is a Leibniz algebra. Since $G_T$ and $V$ are isomorphic, so there
  is an induced Leibniz algebra structure on $V$. This Leibniz algebra
  structure on $V$ is exactly the one given in Proposition
  \ref{average-Leibniz}.
\end{rmk}

In the rest of this section, we give various interesting examples.

\begin{ex}[differential Lie algebras]\label{example-1}{\rm
Let $(\g,[\cdot,\cdot]_\g,d)$ be a differential Lie algebra, 
  that is a Lie algebra $(\g,[\cdot,\cdot]_\g)$ with a derivation $d$
  such that $d\circ d=0$. Then we have
  $$d[dx,y]_\g=[d^2x,y]_\g+[dx,dy]_\g=[dx,dy]_\g.  $$
 Thus, $d$ is an embedding tensor on the \LR
  pair $((\g,[\cdot,\cdot]_\g),(\g;\ad))$.}
\end{ex}

\begin{ex}[an example from supergravity]\label{ex:E8}
{\rm This example is taken from physics literature \cite{KS,Sam}  on supergravity in space-time dimension 4, which is
  one of the origins where the concept of embedding tensors
  appear. The vector space $V$ is taken as the fundamental
  representation ${\bf 56}=(\wedge^2 \Real^8) \oplus (\wedge^2 \Real^8)^*$,
  of $E_{7(7)}$, the non-compact real form of $E_7$. We take
  $\g=\so(8)$, the Lie algebra of real skew-symmetric matrices. In fact
  ${\rm SO}(8)\subset E_{7(7)}$ and the naturally induced representation $\rho$ of
  $\so(8)$ on $V$ is simply the sum of a wedge product of the
  fundamental representation of $\so(8)$ and its dual. More precisely, $\so(8)$ naturally represents on $W:=\wedge^2\mathbb R^8$ via
  $$
  \bar{\rho}(A)(u\wedge v)=(Au)\wedge v+u\wedge Av,\quad \forall u,v\in\mathbb R^8.
  $$
  Let $\bar{\rho}^*$ be the dual representation of $\so(8)$ on $W^*$. Then $\rho=\bar{\rho}+\bar{\rho}^*$ is a representation of $\so(8)$ on $V.$
Let $E_{ij}=R_{ij}-R_{ji}$ be a basis of $\g$, where $R_{ij}$ denotes the $8\times 8$ matrix with the $(i,j)$-position being 1 and elsewhere being 0. Then we have
   \begin{eqnarray*}
     [E_{ij},E_{kl}]=\delta_{jk}E_{il}-\delta_{ik}E_{jl}+ \delta_{jl}E_{ki}-\delta_{il}E_{kj}.
   \end{eqnarray*}
Let $\{e_1,\cdots,e_8\}$ be the basis of $\mathbb R^8$ where $e_i$ is the vector with the $i$-th position being $1$ and elsewhere being $0$. Then $\{e_i\wedge e_j\}_{i<j}$ forms a basis of $W$. Let $\{e_1^*,\cdots,e_8^*\}$ be the dual basis. So $\{e_i^*\wedge e_j^*\}_{i<j}$ forms a basis of $W^*.$

  Define $T:V\lon\so(8)$ by
  $$
  T(e_i\wedge e_j)=E_{ij}, \quad T(e_i^*\wedge e_j^*)=0.
  $$
  Then we have
\begin{eqnarray*}
  T(\rho(T(e_i\wedge e_j))(e_k\wedge e_l))&=&T((E_{ij}e_k)\wedge e_l+e_k\wedge E_{ij}e_l)\\
  &=&T(\delta_{jk}e_i\wedge e_l-\delta_{ik}e_j\wedge e_l+ \delta_{jl}e_k\wedge e_i-\delta_{il}e_k\wedge e_j)\\
  &=&\delta_{jk}E_{il}-\delta_{ik}E_{jl}+ \delta_{jl}E_{ki}-\delta_{il}E_{kj}\\
  &=&[E_{ij},E_{kl}]\\
  &=&[T(e_i\wedge e_j),T(e_k\wedge e_l)],
\end{eqnarray*}
which implies that $T$ is an embedding tensor on the \LR pair $(\so(8),V)$.

 The induced Leibniz algebra structure on $V$  is given by
\begin{eqnarray*}
 [e_i\wedge e_j+e_p^*\wedge e_q^*,e_k\wedge e_l+e_m^*\wedge e_n^*]_T&=&\rho(T(e_i\wedge e_j+e_p^*\wedge e_q^*))(e_k\wedge e_l+e_m^*\wedge e_n^*)\\
 &=&\bar{\rho}(E_{ij})(e_k\wedge e_l)+\bar{\rho}^*(E_{ij})(e_m^*\wedge e_n^*)\\
 &=&\delta_{jk}e_i\wedge e_l-\delta_{ik}e_j\wedge e_l+\delta_{jl}e_k\wedge e_i-\delta_{il}e_k\wedge e_j\\
 &&+\delta_{jm}e_i^*\wedge e_n^*-\delta_{im}e_j^*\wedge e_n^*+\delta_{jn}e_m^*\wedge e_i^*-\delta_{in}e_m^*\wedge e_j^*.
\end{eqnarray*}
Notice that even though the first term $\bar{\rho}(E_{ij})(e_k\wedge
e_l)=[E_{ij}, E_{kl}]=-\bar{\rho}(E_{kj})(e_i\wedge e_j)$ has
antisymmetric property, the second term make the bracket $[\cdot,\cdot]_T$ not
antisymmetric. Thus we end up really with a Leibniz algebra,  not a
Lie algebra. Mathematically, this example can be generalized to all
Lie algebra $\g$ acts on $V=\g \oplus \g^*$ via its adjoint and coadjoint representation.
That is, the natural
projection to the first factor $T: V\to \g$ is an embedding tensor on $\g$ with respect to the action on $V=\g \oplus \g^*$.
}
\end{ex}

\begin{ex}[endomorphism algebra of a 2-term complex]\label{example-3}{\rm
Given a 2-term complex of vector spaces $W\stackrel{\frkT}{\lon}\h$, we
define $\End (W\stackrel{\frkT}{\lon}\h)$ by
$$
\End (W\stackrel{\frkT}{\lon}\h)\triangleq\{(A_0,A_1)\in\gl(\h)\oplus
\gl(W)|A_0\circ\frkT=\frkT\circ A_1\}.
$$
 It is obvious that $\End (W\stackrel{\frkT}{\lon}\h)$ with the commutator bracket $[\cdot,\cdot]_C$ is a Lie algebra. Moreover it represents on $\Hom(\h,W)$ via
 $$
 \rho(A_0,A_1)(\Phi)=[(A_0,A_1),\Phi]_C\triangleq A_1\circ \Phi-\Phi\circ A_0,\quad \forall (A_0,A_1)\in \End (W\stackrel{\frkT}{\lon}\h),~\Phi\in \Hom(\h,W).
 $$
Define $\overline{\frkT}:\Hom(\h,W)\longrightarrow
\End (W\stackrel{\frkT}{\lon}\h)$  by
$$
\overline{\frkT}(\Phi)\triangleq(\Phi\circ\frkT,\frkT\circ\Phi),\quad\forall~\Phi\in\Hom(\h,W).
$$
Then it is straightforward to deduce that $(\End (W\stackrel{\frkT}{\lon}\h),\Hom(\h,W),\overline{\frkT})$ is a Lie-Leibniz triple.

}

\end{ex}

This Lie-Leibniz triple plays an important role in the representation
theory of Lie-Leibniz triples. See Definition \ref{defi:RepLieLei} for
more details. In fact, this embedding tensor comes from a strict 2-Lie
algebra structure on the endomorphisms  of a 2-term complex
 described in
\cite{shengzhu2}. This can be generalized to any strict Lie 2-algebra as follows.

\begin{ex}[strict Lie 2-algebras]\label{example-2}{\rm
A strict Lie $2$-algebra 
is a  $2$-term  graded vector spaces $ \g= \g_{1}\oplus
  \g_{0}$ equipped  with a linear map $\dM_\g:\g_1\longrightarrow \g_0$ and a skew-symmetric bilinear map $[\cdot,\cdot]_\g:\g_i\wedge\g_j\longrightarrow \g_{i+j},0\le i+j\le 1$, such that for all $x,y,z\in \g_0,~a,b\in \g_1$   the following equalities are satisfied:
\begin{itemize}
\item[\rm(a)] $\dM_{\g} [x,a]_\g=[x,\dM_{\g} a]_\g$, \quad $[\dM_{\g} a,b]_\g=[a,\dM_{\g} b]_\g$,
\item[\rm(b)]$[[x,y]_\g,z]_\g+[[y,z]_\g,x]_\g+[[z,x]_\g,y]_\g=0$,\quad $[[x,y]_\g,a]_\g+[[y,a]_\g,x]_\g+[[a,x]_\g,y]_\g=0$.

\end{itemize}
Define $\rho$ from $\g_{0}$ to $\gl(\g_{1})$ by
$
\rho(x)(a)=[x,a]_\g.
$
Then, $(\g_{1};\rho)$ is a representation of the Lie algebra $(\g_0,[\cdot,\cdot]_\g)$ and $\dM_\g$ is an embedding tensor on the \LR pair $((\g_0,[\cdot,\cdot]_\g),(\g_{1};\rho))$.   }
\end{ex}

As strict Lie 2-algebras are equivalent to crossed modules of Lie algebras,
we naturally have the following example.
\begin{ex}[crossed modules of Lie algebras]\label{crossed module}{\rm
A crossed module of Lie algebras is a quadruple $(\g_0,\g_1,\dM,\rho)$, where $\g_0,\g_1$ are Lie algebras, $\dM:\g_1\lon\g_0$  and $\rho:\g_0\lon\Der(\g_1)$ are homomorphisms of Lie algebras such that for all $x\in\g_0$ and $m,n\in\g_1$, we have
  $$\dM(\rho(x)m)=[x,\dM(m)]_{\g_0},\quad
   \rho(\dM(m))(n)=[m,n]_{\g_1}.$$
Then $d$ is an embedding tensor on the \LR pair $((\g_0,[\cdot,\cdot]_{\g_0}),(\g_1;\rho))$.  }
\end{ex}

Example \ref{example-2} and Example \ref{crossed module} can be generalized to the following more general case.

\begin{ex}[Lie objects in the infinitesimal tensor category of linear maps]\label{Lie object111}{\rm
Let $(\g,[\cdot,\cdot]_\g)$ be a Lie algebra and $(V;\rho)$ a representation. If a linear map $T:V\lon\g$ is $\g$-equivariant, that is,
\begin{eqnarray*}
 T(\rho(x)v)=[x,Tv]_\g,\quad \forall x\in\g,~v\in V,
\end{eqnarray*}
then $T$ is an embedding tensor on the \LR pair $((\g,[\cdot,\cdot]_\g),(V;\rho))$. In fact $V\stackrel{T}{\lon}\g$ is a Lie object in the infinitesimal tensor category of linear maps if and only if $T:V\lon\g$ is $\g$-equivariant. Let $(\la,[\cdot,\cdot]_\la)$ be a Leibniz algebra and  $\la^{\ann}$ be the two-sided ideal of $\la$ generated by $[x,x]_\la$ for all $x\in\la$. Then the natural projection $\pi$ from $\la$ to $\la/\la^{\ann}$ gives a Lie object   in the infinitesimal tensor category of linear maps and $\pi$ is an embedding tensor. See \cite{KS,Loday-Pirashvili} for more details.}
\end{ex}

\emptycomment{

 \begin{ex}\label{example-4}{\rm
    Consider the unique $2$-dimensional non-abelian Lie algebra on $\mathbb C^2$ given with respect to a basis $\{e_1,e_2\}$ by
   $[e_1,e_2]=e_1$.
   Then $T=\left(\begin{array}{cc}a_{11}&a_{12}\\
                                                                a_{21}&a_{22}\end{array}\right)$ is an embedding tensor on $\mathbb C^2$ with respect to the adjoint representation if and only if
   \begin{eqnarray*}
   &&[Te_i,Te_j]=T[Te_i,e_j],\quad i,j\in\{1,2\}.
   \end{eqnarray*}
   First by
   $$
 0=[Te_1,Te_1]= T[Te_1,e_1]=T[a_{11}e_1+a_{21}e_2, e_1]=-a_{21}a_{11}e_1-a_{21}^2e_2,
   $$
   we have $a_{21}=0$. Then by
   $$
 0=[Te_2,Te_2]=  T[Te_2,e_2]=T[a_{12}e_1+a_{22}e_2, e_2]=a_{12}a_{11}e_1+a_{12}a_{21}e_2,
   $$
  we obtain $a_{12}a_{11}=0$. Finally by
\begin{eqnarray*}
  [Te_1,Te_2]
  =T[Te_1,e_2]
  =-T[Te_2,e_1],
\end{eqnarray*}
  we have $a_{11}a_{22}=a_{11}^2$. Summarize the above discussion, we have
   \begin{itemize}
     \item[\rm(i)] If $a_{11}=0$, then  any $T=\left(\begin{array}{cc}0&a_{12}\\
   0&a_{22}\end{array}\right)$ is an embedding tensor on $\mathbb C^2$ with respect to the adjoint representation;
     \item[\rm(ii)] If $a_{11}\not=0$, then $a_{12}=0$ and $a_{22}=a_{11}$. Thus, any $T=\left(\begin{array}{cc}a_{11}&0\\
   0&a_{11}\end{array}\right)$ is an embedding tensor on $\mathbb C^2$ with respect to the adjoint representation.
   \end{itemize}
   }
    \end{ex}
    }

\begin{ex}\label{example-5}{\rm
The {\bf Heisenberg algebra} $H_3(\mathbb C)$ is a
three-dimensional  complex  Lie algebra generated by
elements $e_1, e_2$ and $e_3$ with Lie brackets
$
[e_1,e_2]=e_3,~ [e_1,e_3]=0,~  [e_2,e_3]=0.
$
  $T=\left(\begin{array}{ccc}r_{11}&r_{12}&r_{13}\\
r_{21}&r_{22}&r_{23}\\
r_{31}&r_{32}&r_{33}\end{array}\right)$ is an embedding tensor on $H_3(\mathbb C)$ with respect to the adjoint representation  if and only if
$
   [Te_i,Te_j]=T[Te_i,e_j],~~ i,j\in\{1,2,3\}.
$
First by
\begin{eqnarray*}
 0=[Te_1,Te_1]= T[Te_1,e_1]=T[r_{11}e_1+r_{21}e_2+r_{31}e_3, e_1]=-r_{21}r_{13}e_1-r_{21}r_{23}e_2-r_{21}r_{33}e_3,
\end{eqnarray*}
we obtain
$
r_{21}r_{13}=0,~ r_{21}r_{23}=0,~ r_{21}r_{33}=0.
$
Similarly, we can obtain
$$\begin{array}{rcl  rcl rcl}
  r_{12}r_{13}&=&0, & r_{12}r_{23}&=&0, &  r_{12}r_{33}&=&0,\\
 r_{11}r_{13}&=&r_{22}r_{13}=0,&   r_{11}r_{23}&=&r_{22}r_{23}=0, & r_{11}r_{22}-r_{12}r_{21}&=&r_{11}r_{33}=r_{22}r_{33},\\
 r_{23}r_{13}&=&0, & r_{23}^2&=&0, & r_{23}r_{33}&=&r_{11}r_{23}-r_{13}r_{21}=0,\\
  r_{13}^2&=&0, & r_{13}r_{23}&=&0, & -r_{13}r_{33}&=&r_{12}r_{23}-r_{13}r_{22}=0.
\end{array}
$$
\emptycomment{Then by
\begin{eqnarray*}
&&[Te_1,Te_2]=[r_{11}e_1+r_{21}e_2+r_{31}e_3,r_{12}e_1+r_{22}e_2+r_{32}e_3]=(r_{11}r_{22}-r_{12}r_{21})e_3\\
&&T[Te_1,e_2]=T[r_{11}e_1+r_{21}e_2+r_{31}e_3,e_2]=r_{11}r_{13}e_1+r_{11}r_{23}e_2+r_{11}r_{33}e_3\\
&&T[e_1,Te_2]=T[e_1,r_{12}e_1+r_{22}e_2+r_{32}e_3]=r_{22}r_{13}e_1+r_{22}r_{23}e_2+r_{22}r_{33}e_3\\
&&[Te_1,Te_3]=[r_{11}e_1+r_{21}e_2+r_{31}e_3,r_{13}e_1+r_{23}e_2+r_{33}e_3]=(r_{11}r_{23}-r_{13}r_{21})e_3\\
&&T[Te_1,e_3]=0\\
&&T[e_1,Te_3]=T[e_1,r_{13}e_1+r_{23}e_2+r_{33}e_3]=r_{23}r_{13}e_1+r_{23}r_{23}e_2+r_{23}r_{33}e_3\\
&&[Te_2,Te_3]=[r_{12}e_1+r_{22}e_2+r_{32}e_3,r_{13}e_1+r_{23}e_2+r_{33}e_3]=(r_{12}r_{23}-r_{13}r_{22})e_3\\
&&T[Te_2,e_3]=0\\
&&T[e_2,Te_3]=T[e_2,r_{13}e_1+r_{23}e_2+r_{33}e_3]=-r_{13}^2e_1-r_{13}r_{23}e_2-r_{13}r_{33}e_3.
\end{eqnarray*}
}
Therefore, we have
   \begin{itemize}
     \item[\rm(i)] If $r_{13}=r_{23}=r_{33}=0$,  then   $T=\left(\begin{array}{ccc}r_{11}&r_{12}&0\\
r_{21}&r_{22}&0\\
r_{31}&r_{32}&0\end{array}\right)$ is an embedding tensor on $H_3(\mathbb C)$ if and only if $r_{11}r_{22}=r_{12}r_{21}$.
     \item[\rm(ii)] If $r_{13}=r_{23}=0$ and $r_{33}\not=0$, then $r_{12}=r_{21}=0$ and    $T=\left(\begin{array}{ccc}r_{11}&0&0\\
0&r_{22}&0\\
r_{31}&r_{32}&r_{33}\end{array}\right)$ is an embedding tensor on $H_3(\mathbb C)$ if and only if   $r_{11}r_{22}=r_{22}r_{33}=r_{11}r_{33}$.
   \end{itemize}
These two cases exhaust all the possibilities of embedding tensors on
$H_3(\mathbb C)$ with respect to the adjoint representation.
   }
\end{ex}

\emptycomment{
\begin{ex}\label{example-3}
The $sl_2(\mathbb C)$ is the
three-dimensional  complex simple Lie algebra generated by
elements $e,f$ and $h$ with Lie brackets
\begin{eqnarray*}
[e,h]=2e,\quad [f,h]=-2f,\quad  [e,f]=h.
\end{eqnarray*}
For all $n\ge 0$, there exists one irreducible representation $(V_n;\rho_n)$ of $sl_2(\mathbb C)$ (up to isomorphism) of dimension $n+1$.  The vector space $V_n$ with a basis $\{v_0,v_1,\cdots,v_n\}$ and the representation of $sl_2(\mathbb C)$ on $V_n$ is given by
\begin{eqnarray*}
\rho_n(h)v_i=(n-2i)v_i,\quad\rho_n(e)v_i=(n-i+1)v_{i-1},\quad\rho_n(f)v_i=(i+1)v_{i+1},
\end{eqnarray*}
here we set $v_{-1}=v_{n+1}=0.$ Work out the average operators with respect to representation $(V_n;\rho_n)$ will be sophisticated unless we get help from the computers. However, for $n=1$
$$
   Tv_0=r_{11}e+r_{21}f+r_{31}h,\quad Tv_1=r_{12}e+r_{22}f+r_{32}h.
   $$
It is straightforward to see that $T$ is an average operator with respect to representation $(V_1;\rho_1)$ if and only if
$$
?
$$
For $n=2$, we have
$$
   Tv_0=r_{11}e+r_{21}f+r_{31}h,\quad Tv_1=r_{12}e+r_{22}f+r_{32}h,\quad Tv_2=r_{13}e+r_{23}f+r_{33}h.
$$
It is straightforward to see that $T$ is an average operator with respect to representation $(V_2;\rho_2)$\footnote{Since $sl_2(\mathbb C)$ is a complex simple Lie algebra. The representation $(V_2;\rho_2)$ is isomorphic to the adjoin representation $(sl_2(\mathbb C);\ad)$.} if and only if
$$
?
$$
\end{ex}
}

\section{The controlling graded Lie algebra  and cohomology of embedding tensors}\label{sec:A}

In this section first we recall the controlling g.l.a. that characterize Leibniz algebras as Maurer-Cartan elements  and   the g.l.a. governing  a \LR pair that was originally given in \cite{Arnal}. Then we construct the controlling g.l.a. of  embedding tensors. Finally we introduce the cohomologies of embedding tensors.

\subsection{The controlling graded Lie algebra of \LR pairs}

A permutation $\sigma\in\mathbb S_n$ is called an {\bf $(i,n-i)$-shuffle} if $\sigma(1)<\cdots <\sigma(i)$ and $\sigma(i+1)<\cdots <\sigma(n)$. If $i=0$ or $n$, we assume $\sigma=\Id$. The set of all $(i,n-i)$-shuffles will be denoted by $\mathbb S_{(i,n-i)}$. The notion of an $(i_1,\cdots,i_k)$-shuffle and the set $\mathbb S_{(i_1,\cdots,i_k)}$ are defined analogously.

Let $\g$ be a vector space. We consider the graded vector space $$C^\bullet(\g,\g)=\oplus_{n=0}^{+\infty}C^n(\g,\g)=\oplus_{n=0}^{+\infty}\Hom(\otimes^{n+1}\g,\g).$$ It is known that $C^\bullet(\g,\g)$ equipped with the {\bf  Balavoine bracket} ~\cite{Bal}:
\begin{eqnarray}\label{leibniz-bracket}
[P,Q]_\B=P\bar{\circ}Q-(-1)^{pq}Q\bar{\circ}P,\,\,\,\,\forall P\in C^{p}(\g,\g),Q\in C^{q}(\g,\g),
\end{eqnarray}
is a g.l.a.  where $P\bar{\circ}Q\in C^{p+q}(\g,\g)$ is defined by
$
P\bar{\circ}Q=\sum_{k=1}^{p+1}P\circ_k Q,
$
and $\circ_k$ is defined by
\begin{eqnarray*}
 \nonumber&&(P\circ_kQ)(x_1,\cdots,x_{p+q+1})\\
&=&\sum_{\sigma\in\mathbb S_{(k-1,q)}}(-1)^{(k-1)q}(-1)^{\sigma}P(x_{\sigma(1)},\cdots,x_{\sigma(k-1)},Q(x_{\sigma(k)},\cdots,x_{\sigma(k+q-1)},x_{k+q}),x_{k+q+1},\cdots,x_{p+q+1}).
\end{eqnarray*}

\emptycomment{
\begin{rmk}\label{NR-subalgebra}
Let $\g$ be a vector space. Then $\big(\oplus_{n=0}^{+\infty}\Hom(\wedge^{n+1}\g,\g),[\cdot,\cdot]_\NR\big)$ is a subalgebra of the graded Lie algebra $(\oplus_{n=0}^{+\infty}\Hom(\otimes^{n+1}\g,\g),[\cdot,\cdot]_\B)$, where  $[\cdot,\cdot]_\NR$ is the Nijenhuis-Richardson bracket \cite{NR,NR2}.
\end{rmk}
}

\begin{rmk}\label{B-coder}
In fact, the Balavoine bracket is the commutator of  coderivations on the  cofree conilpotent coZinbiel coalgebra $\bar{\Ten}(\g)$. See \cite{ammardefiLeibnizalgebra,Uchino-1} for more details. Note  that on the same graded vector space $C^\bullet(\g,\g)$ there is the Gerstenhaber bracket \cite{stasheff-93} which is the commutator of  coderivations on the  cofree conilpotent coassociative coalgebra $\bar{\Ten}(\g)$.
\end{rmk}


The following conclusion is straightforward.
\begin{lem}\label{lem:LeibnizMC}
For $\Omega\in C^{1}(\g,\g)=\Hom(\otimes^2\g,\g)$, we have
\begin{eqnarray*}
\half[\Omega,\Omega]_\B(x_1,x_2,x_3)=\Omega\bar{\circ}\Omega(x_1,x_2,x_3)=\Omega(\Omega(x_1,x_2),x_3)-\Omega(x_1,\Omega(x_2,x_3))
+\Omega(x_2,\Omega(x_1,x_3)).
\end{eqnarray*}
Thus, $\Omega$ defines a Leibniz algebra structure if and only if $[\Omega,\Omega]_\B=0$, i.e. $\Omega$ is a Maurer-Cartan element of the g.l.a. $( C^\bullet(\g,\g),[\cdot,\cdot]_\B)$.
\end{lem}

Let $\g_1$ and $\g_2$ be vector spaces and elements in $\g_1$ will be denoted by $x,y,z, x_i$ and elements in $\g_2$ will be denoted by $u,v,v_i$.
Let $c:\g_2^{\otimes n}\lon \g_1$ be a linear map. Define a linear map $\hat{c}\in C^{n-1}(\g_1\oplus\g_2,\g_1\oplus\g_2)$ by
\begin{eqnarray*}
\hat{c}\big((x_1,v_1)\otimes\cdots\otimes(x_n,v_n)\big):=(c(v_1,\cdots,v_n),0).
\end{eqnarray*}
In general, for a given linear map $f:\g_{i(1)}\otimes\g_{i(2)}\otimes\cdots\otimes\g_{i(n)}\lon\g_j$, $i(1),\cdots,i(n),j\in\{1,2\}$, we define a linear map $\hat{f}\in C^{n-1}(\g_1\oplus\g_2,\g_1\oplus\g_2)$ by

\[
\hat{f}:=\left\{
\begin{array}{ll}
f &\mbox {on $\g_{i(1)}\otimes\g_{i(2)}\otimes\cdots\otimes\g_{i(n)}$, }\\
0 &\mbox {all other cases.}
\end{array}
\right.
\]
We call the linear map $\hat{f}$ a {\bf horizontal lift} of $f$. 

\emptycomment{
We denote by $\g^{l,k}$ the direct sum of all $(l+k)$-tensor powers of $\g_1$ and $\g_2$, where $l$ (resp. $k$) is the number of $\g_1$ (resp. $\g_2$).
By the properties of the $\Hom$-functor, we have
\begin{eqnarray}\label{decomposition}
C^{n-1}(\g_1\oplus\g_2,\g_1\oplus\g_2)\cong\sum_{l+k=n}\Hom(\g^{l,k},\g_1)\oplus\sum_{l+k=n}\Hom(\g^{l,k},\g_2),
\end{eqnarray}
where the isomorphism is the horizontal lift.

\begin{defi}{\rm (\cite{ST})}
A linear map $f\in \Hom\big(\otimes^n(\g_1\oplus\g_2),(\g_1\oplus\g_2)\big)$ has a {\bf bidegree} $l|k$, which is denoted by $||f||=l|k$,   if $f$ satisfies the following four conditions:
\begin{itemize}
\item[\rm(i)] $l+k+1=n;$
\item[\rm(ii)] If $X$ is an element in $\g^{l+1,k}$, then $f(X)\in\g_1;$
\item[\rm(iii)] If $X$ is an element in $\g^{l,k+1}$, then $f(X)\in\g_2;$
\item[\rm(iv)] All the other case, $f(X)=0.$
\end{itemize}
\end{defi}
A linear map $f$ is said to be homogeneous  if $f$ has a bidegree.
We have $l+k\ge0,~k,l\ge-1$ because $n\ge1$ and $l+1,~k+1\ge0$. For instance, the lift $\hat{H}\in C^0(\g_1\oplus\g_2,\g_1\oplus\g_2)$ of $H:\g_2\lon\g_1$ has the bidegree $-1|1$.

It is obvious that we have the following lemmas:

\begin{lem}\label{Zero-condition-2}
If $||f||=-1|l$ (resp. $l|-1$) and $||g||=-1|k$ (resp. $k|-1$), then $[f,g]_\B=0.$
\end{lem}

\begin{lem}\label{important-lemma-2}
If $||f||=l_f|k_f$ and $||g||=l_g|k_g$, then $[f,g]_\B$ has the bidegree $l_f+l_g|k_f+k_g.$
\end{lem}
}
Let us write $\hat{C}^n(\g\oplus V,\g\oplus V):=\Hom(\wedge^{n+1}\g,\g)\oplus \Hom(\wedge ^{n}\g\otimes V,V)$. Using the horizontal lift, we can regard $\hat{C}^n(\g\oplus V,\g\oplus V)$ as a subspace of ${C}^n(\g\oplus V,\g\oplus V)$.

\begin{thm}\label{lem:MC-algrep}
The above defined  $\oplus_{n=0}^{+\infty}  \hat{C}^n(\g\oplus V,\g\oplus V)$ is a subalgebra of the g.l.a. $({C}^\bullet(\g\oplus V,\g\oplus V),[\cdot,\cdot]_\B)$. Moreover, a Maurer-Cartan element  of the g.l.a. $(\oplus_{n=0}^{+\infty}  \hat{C}^n(\g\oplus V,\g\oplus V),[\cdot,\cdot]_\B)$  is exactly a \LR pair.
\end{thm}

\begin{proof}
It is straightforward to deduce  that $\oplus_{n=0}^{+\infty}  \hat{C}^n(\g\oplus V,\g\oplus V)$ is a subalgebra of the graded Lie algebra $(\oplus_{n=0}^{+\infty}{C}^n(\g\oplus V,\g\oplus V),[\cdot,\cdot]_\B)$.
Moreover, let $(\mu,\rho)\in\Hom(\wedge^{2}\g,\g)\oplus \Hom(\g\otimes V,V)$ be a Maurer-Cartan element. Then $[\hat{\mu}+\hat{\rho},\hat{\mu}+\hat{\rho}]_\B=0$    implies that $\mu$ defines a Lie algebra structure on $\g$ and $\rho$ is a representation of the Lie algebra $(\g,\mu)$ on $V.$
\end{proof}

\subsection{The controlling graded Lie algebra of  embedding tensors}
 Let $((\g,[\cdot,\cdot]_{\g}),(V;\rho))$ be a \LR pair. Usually we will also use $\mu$ to indicate the Lie bracket $[\cdot,\cdot]_\g$. We have a Leibniz algebra structure $\mu\boxplus\rho$ on $\g\oplus V$, which is given by
\begin{eqnarray}
(\mu\boxplus\rho)\big((x,u),(y,v)\big)=([x,y]_\g,\rho(x)v).
\end{eqnarray}
This Leibniz algebra is called the {\bf hemisemidirect product} of the Lie algebra $(\g,[\cdot,\cdot]_{\g})$ and the representation  $(V;\rho)$. It first appeared in \cite{Kinyon}.

\begin{thm} \label{average-MC}
Let $((\g,[\cdot,\cdot]_{\g}),(V;\rho))$ be a \LR pair. Then
$(\oplus_{k=1}^{+\infty} \Hom(\otimes^{k}V,\g),\Courant{\cdot,\cdot})$
is a graded Lie algebra, where the graded Lie  bracket
$\Courant{\cdot,\cdot}$ is given in a derived fashion
\begin{equation}\label{eq:glaO}
\Courant{\theta,\phi}=(-1)^{m-1}[[\mu\boxplus\rho,\theta]_\B,\phi]_\B,\quad \forall \theta\in \Hom(\otimes^{m}V,\g),~\phi \in \Hom(\otimes^{n}V,\g).
\end{equation}
More precisely, it is given by
{\footnotesize\begin{eqnarray*}
&&\nonumber\Courant{\theta,\phi }(v_1,\cdots,v_{m+n})\\
&=&\sum_{k=1}^{m}\sum_{\sigma\in\mathbb S_{(k-1,n)}}(-1)^{(k-1)n+1}(-1)^{\sigma}\theta(v_{\sigma(1)},\cdots,v_{\sigma(k-1)},\rho(\phi (v_{\sigma(k)},\cdots,v_{\sigma(k+n-1)}))v_{k+n},v_{k+n+1},\cdots,v_{m+n})\\
&&+\sum_{\sigma\in\mathbb S_{(m,n)}}(-1)^{mn+1}(-1)^{\sigma}[\theta(v_{\sigma(1)},\cdots,v_{\sigma(m)}),\phi (v_{\sigma(m+1)},\cdots,v_{\sigma(m+n-1)},v_{\sigma(m+n)})]_{\g}\\
&&+\sum_{k=1}^{n}\sum_{\sigma\in\mathbb S_{(k-1,m)}}(-1)^{m(k+n-1)}(-1)^{\sigma}\phi (v_{\sigma(1)},\cdots,v_{\sigma(k-1)},\rho(\theta(v_{\sigma(k)},\cdots,v_{\sigma(k+m-1)}))v_{k+m},v_{k+m+1},\cdots,v_{m+n}).
\end{eqnarray*}
}

Moreover, its Maurer-Cartan elements are precisely  embedding tensors.
\end{thm}
\begin{proof}
 In short, the graded Lie algebra $(\oplus_{k=1}^{+\infty} \Hom(\otimes^{k}V,\g),\Courant{\cdot,\cdot})$ is obtained via the derived bracket \cite{Kosmann-Schwarzbach,Kosmann-Schwarzbach-1,Vo}. In fact, the Balavoine bracket   $[\cdot,\cdot]_{\B}$ associated to the direct sum vector space $\g\oplus V$ gives rise to a graded Lie algebra $(\oplus_{k= 1}^{+\infty}\Hom(\otimes^k(\g\oplus V),\g\oplus V),[\cdot,\cdot]_{\B})$. 
 Since $\mu:\wedge^2\g\longrightarrow \g$ is a Lie algebra structure and $\rho:\g\otimes V\longrightarrow V$   is a representation of $\g$ on $V$, therefore the  hemisemidirect product Leibniz algebra structure $\mu\boxplus\rho$ is a Maurer-Cartan element of the graded Lie algebra $(\oplus_{k=1}^{+\infty}\Hom(\otimes^k(\g\oplus V),\g\oplus V),[\cdot,\cdot]_{\B})$, defining a differential $d_{\mu\boxplus\rho}$ on $(\oplus_{k=1}^{+\infty}\Hom(\otimes^k(\g\oplus V),\g\oplus V),[\cdot,\cdot]_{\B})$ via
$
  d_{\mu\boxplus\rho}=[\mu\boxplus\rho,\cdot]_{\B}.
$
Since the subspace $\oplus_{k=1}^{+\infty}\Hom(\otimes^kV,\g)$ is
abelian under $[\cdot, \cdot]_\B$ by degree reasons, 
 the differential $d_{\mu\boxplus\rho}$ gives rise to a graded Lie algebra structure on the graded vector space $\oplus_{k=1}^{+\infty}\Hom(\otimes^kV,\g)$ via the derived bracket
\eqref{eq:glaO}.  

For $T\in \Hom(V,\g)$, we have
\begin{eqnarray*}
\half\Courant{T,T}(v_1,v_2)= [Tv_1,Tv_2]_\g-T(\rho(Tv_1)v_2),\quad\forall v_1,v_2\in V,
\end{eqnarray*}
which implies that Maurer-Cartan elements are precisely  embedding tensors.
 \end{proof}

There is a close relationship between the graded Lie algebra $(\oplus_{k=1}^{+\infty} \Hom(\otimes^{k}V,\g),\Courant{\cdot,\cdot})$ and the graded Lie algebra $(C^\bullet(V,V),[\cdot,\cdot]_\B)$, where $[\cdot,\cdot]_\B$ is the Balavoine bracket  defined by \eqref{leibniz-bracket}.

 Define a linear map $\Phi:\Hom(\otimes^kV,\g)\lon\Hom(\otimes^{k+1}V,V)$ for $k=1,2,\cdots,$   by
\begin{eqnarray}\label{eq:phi}
\qquad\Phi(f)(u_1,\cdots,u_k,u_{k+1})=-\rho(f(u_1,\cdots,u_k))u_{k+1},\,\,\,\,\forall f\in\Hom(\otimes^kV,\g),~ u_1,\cdots, u_{k+1}\in V.
\label{eq:defiphi}
\end{eqnarray}


\begin{pro}\label{pro:homoglaop}
The linear map $\Phi$, defined by  \eqref{eq:phi}, is a
homomorphism of graded Lie algebras from
$(\oplus_{k=1}^{+\infty} \Hom(\otimes^{k}V,\g),\Courant{\cdot,\cdot})$ to
$(C^\bullet(V,V),[\cdot,\cdot]_\B)$.
\end{pro}

\begin{proof}
For $\theta\in\Hom(\otimes^mV,\g)$ and $\phi \in\Hom(\otimes^nV,\g)$, we have $[\Phi(\theta),\Phi(\phi )]_\B\in\Hom(\otimes^{m+n+1}V,V)$. More precisely, for all $u_1,\cdots,u_{m+n+1}\in V$, we have
{\footnotesize\begin{eqnarray*}
&&\Phi(\theta)\bar{\circ}\Phi(\phi )(u_1,\cdots,u_{m+n+1})\\
&=&\sum_{k=1}^{m+1}\sum_{\sigma\in\mathbb S_{(k-1,n)}}(-1)^{(k-1)n}(-1)^{\sigma}\Phi(\theta)(u_{\sigma(1)},\cdots,u_{\sigma(k-1)},\Phi(\phi )(u_{\sigma(k)},\cdots,u_{\sigma(k+n-1)},u_{k+n}),u_{k+n+1},\cdots,u_{m+n+1})\\
&=&\sum_{k=1}^{m}\sum_{\sigma\in\mathbb S_{(k-1,n)}}(-1)^{(k-1)n}(-1)^{\sigma}\rho(\theta(u_{\sigma(1)},\cdots,u_{\sigma(k-1)},\rho(\phi (u_{\sigma(k)},\cdots,u_{\sigma(k+n-1)}))u_{k+n},u_{k+n+1},\cdots,u_{m+n}))u_{m+n+1}\\
&&+\sum_{\sigma\in\mathbb S_{(m,n)}}(-1)^{mn}(-1)^{\sigma}\rho(\theta(u_{\sigma(1)},\cdots,u_{\sigma(m)}))\rho(\phi (u_{\sigma(m+1)},\cdots,u_{\sigma(m+n)}))u_{m+n+1}.
\end{eqnarray*}
}
For any $\tau\in\mathbb S_{(n,m)}$, we can define $\sigma\in\mathbb S_{(m,n)}$ by
\[
\sigma(i)=\left\{
\begin{array}{ll}
\tau(n+i), & 1\le i\le m;\\
\tau(i-m), & m+1\le i\le m+n.
\end{array}
\right.
\]
Thus, we have $(-1)^{\sigma}=(-1)^{mn}(-1)^{\tau}$.
In fact, the elements of $\mathbb S_{(n,m)}$ are in one-to-one correspondence with the elements of $\mathbb S_{(m,n)}$. Then we have
{\footnotesize
\begin{eqnarray*}
&&\Phi(\phi )\bar{\circ}\Phi(\theta)(u_1,\cdots,u_{m+n+1})\\
&=&\sum_{k=1}^{n}\sum_{\tau\in\mathbb S_{(k-1,m)}}(-1)^{(k-1)m}(-1)^{\tau}\rho(\phi (u_{\tau(1)},\cdots,u_{\tau(k-1)},\rho(\theta(u_{\tau(k)},\cdots,u_{\tau(k+m-1)}))u_{k+m},u_{k+m+1},\cdots,u_{m+n}))u_{m+n+1}\\
&&+\sum_{\tau\in\mathbb S_{(n,m)}}(-1)^{mn}(-1)^{\tau}\rho(\phi (u_{\tau(1)},\cdots,u_{\tau(n)}))\rho(\theta(u_{\tau(n+1)},\cdots,u_{\tau(m+n)}))u_{m+n+1}\\
&=&\sum_{k=1}^{n}\sum_{\sigma\in\mathbb S_{(k-1,m)}}(-1)^{(k-1)m}(-1)^{\sigma}\rho(\phi (u_{\sigma(1)},\cdots,u_{\sigma(k-1)},\rho(\theta(u_{\sigma(k)},\cdots,u_{\sigma(k+m-1)}))u_{k+m},u_{k+m+1},\cdots,u_{m+n}))u_{m+n+1}\\
&&+\sum_{\sigma\in\mathbb S_{(m,n)}}(-1)^{\sigma}\rho(\phi (u_{\sigma(m+1)},\cdots,u_{\sigma(m+n)}))\rho(\theta(u_{\sigma(1)},\cdots,u_{\sigma(m)}))u_{m+n+1}.
\end{eqnarray*}
}
Therefore, we have
\begin{eqnarray*}
&&[\Phi(\theta),\Phi(\phi )]_\B(u_1,\cdots,u_{m+n+1})\\
&=&\sum_{k=1}^{m}\sum_{\sigma\in\mathbb S_{(k-1,n)}}(-1)^{(k-1)n}(-1)^{\sigma}\\
&&\rho(\theta(u_{\sigma(1)},\cdots,u_{\sigma(k-1)},\rho(\phi (u_{\sigma(k)},\cdots,u_{\sigma(k+n-1)}))u_{k+n},u_{k+n+1},\cdots,u_{m+n}))u_{m+n+1}\\
&&-(-1)^{mn}\sum_{k=1}^{n}\sum_{\sigma\in\mathbb S_{(k-1,m)}}(-1)^{(k-1)m}(-1)^{\sigma}\\
&&\rho(\phi (u_{\sigma(1)},\cdots,u_{\sigma(k-1)},\rho(\theta(u_{\sigma(k)},\cdots,u_{\sigma(k+m-1)}))u_{k+m},u_{k+m+1},\cdots,u_{m+n}))u_{m+n+1}\\
&&+(-1)^{mn}\sum_{\sigma\in\mathbb S_{(m,n)}}(-1)^{\sigma}\rho([\theta(u_{\sigma(1)},\cdots,u_{\sigma(m)}),\phi (u_{\sigma(m+1)},\cdots,u_{\sigma(m+n)})]_\g)u_{m+n+1}\\
&=&\Phi(\Courant{\theta,\phi })(u_1,\cdots,u_{m+n+1}).
\end{eqnarray*}
Thus, $\Phi$ is a homomorphism  from $(\oplus_{k=1}^{+\infty} \Hom(\otimes^{k}V,\g),\Courant{\cdot,\cdot})$ to
$(C^\bullet(V,V),[\cdot,\cdot]_\B)$.
\end{proof}

\begin{rmk}
  The above result gives another intrinsic explanation of Proposition \ref{average-Leibniz}. By Lemma \ref{lem:LeibnizMC}, Maurer-Cartan elements of the g.l.a
$(C^ \bullet(V,V),[\cdot,\cdot]_\B)$ are Leibniz algebra structures on $V$. By Theorem \ref{average-MC}, Maurer-Cartan elements of the g.l.a    $(\oplus_{k=1}^{+\infty} \Hom(\otimes^{k}V,\g),\Courant{\cdot,\cdot})$ are  embedding tensors. By Proposition \ref{pro:homoglaop}, $\Phi$ sends Maurer-Cartan elements to Maurer-Cartan elements. Thus, an  embedding tensor $T:V\lon\g$ induces a Leibniz algebra structure on $V$.
\end{rmk}

\subsection{Cohomology of  embedding tensors}

First we review representations and cohomologies of Leibniz algebras.
A {\bf representation} of a Leibniz algebra $(\la,[\cdot,\cdot]_{\la})$ is a triple $(V;\rho^L,\rho^R)$, where $V$ is a vector space, $\rho^L,\rho^R:\la\lon\gl(V)$ are linear maps such that  for all $x,y\in\la$,
\begin{eqnarray*}
\rho^L([x,y]_{\la})=[\rho^L(x),\rho^L(y)],\quad
\rho^R([x,y]_{\la})=[\rho^L(x),\rho^R(y)],\quad
\rho^R(y)\circ \rho^L(x)=-\rho^R(y)\circ \rho^R(x).
\end{eqnarray*}
Here $[\cdot,\cdot]:\wedge^2\gl(V)\lon\gl(V)$ is the commutator Lie bracket on $\gl(V)$.
It is straightforward to see that $(\la;L,R)$, where the left multiplication $L:\la\longrightarrow\gl(\la)$ is given by \eqref{eq:lm} and the right multiplication $R:\la\longrightarrow\gl(\la)$ is defined by $R_xy=[y,x]_\la$,
 is a representation of $(\la,[\cdot,\cdot]_{\la})$, which is called the {\bf regular representation}.

\begin{defi}{\rm (\cite{Loday and Pirashvili})}  Let $(V;\rho^L,\rho^R)$ be a representation of a Leibniz algebra $(\la,[\cdot,\cdot]_{\la})$.
The {\bf  Loday-Pirashvili cohomology} of $\la$ with   coefficients in $V$ is the cohomology of the cochain complex $(C^\bullet(\la,V)=\oplus_{k=0}^{+\infty}C^k(\la,V),\partial)$, where $C^k(\la,V)=
\Hom(\otimes^k\la,V)$ and the coboundary operator
$\partial:C^k(\la,V)\longrightarrow C^
{k+1}(\la,V)$
is defined by
\begin{eqnarray*}
(\partial f)(x_1,\cdots,x_{k+1})&=&\sum_{i=1}^{k}(-1)^{i+1}\rho^L(x_i)f(x_1,\cdots,\hat{x_i},\cdots,x_{k+1})+(-1)^{k+1}\rho^R(x_{k+1})f(x_1,\cdots,x_{k})\\
                      \nonumber&&+\sum_{1\le i<j\le k+1}(-1)^if(x_1,\cdots,\hat{x_i},\cdots,x_{j-1},[x_i,x_j]_\la,x_{j+1},\cdots,x_{k+1}),
\end{eqnarray*}
for all $x_1,\cdots, x_{k+1}\in\la$. The resulting cohomology is denoted by $H^*(\la,V)$.
\end{defi}
The regular representation $(\la;L,R)$ is very important. We denote the corresponding cochain complex by $(C^\bullet\la,\la),\partial^{reg})$ and the resulting cohomology   by $H^*_{\rm reg}(\la)$.

 Let $T:V\longrightarrow\g$ be an  embedding tensor on  a \LR pair $((\g,[\cdot,\cdot]_\g),(V;\rho))$.
By Proposition \ref{average-Leibniz}, $[u,v]_T:=\rho(Tu)v$ defines a Leibniz algebra structure on $V$. Furthermore, define $\rho^L:V\lon\gl(\g)$ and $\rho^R:V\lon\gl(\g)$ by
\begin{equation}
\rho^L(u)y:=[Tu,y]_\g,\quad \rho^R(v)x:=[x,Tv]_\g-T(\rho(x)v).
\end{equation}
\begin{lem}
  With above notations, $(\g;\rho^L,\rho^R)$ is  a representation of the Leibniz algebra $(V,[\cdot,\cdot]_T)$.
\end{lem}
\begin{proof}
  It follows from direct verification.
\end{proof}

Let $\partial_T: \Hom(\otimes^kV,\g)\longrightarrow   \Hom(\otimes^{k+1}V,\g)$ be the corresponding Loday-Pirashvili coboundary operator of the Leibniz algebra $(V,[\cdot,\cdot]_T)$ with   coefficients in $(\g;\rho^L,\rho^R)$. More precisely, $\partial_T: \Hom(\otimes^kV,\g)\longrightarrow   \Hom(\otimes^{k+1}V,\g)$ is given by
               \begin{eqnarray}
               \label{eq:defpt} \quad \partial_T \theta(u_1,\cdots,u_{k+1})&=&\sum_{i=1}^{k}(-1)^{i+1}[Tu_i,\theta(u_1,\cdots,\hat{u}_i,\cdots, u_{k+1})]_\g\\
                 \nonumber&&+(-1)^{k+1}[\theta(u_1,\cdots,  u_{k}),Tu_{k+1}]_\g+(-1)^kT(\rho(\theta(u_1,\cdots, u_{k}))u_{k+1})\\
                \nonumber &&+\sum_{1\le i<j\le k+1}(-1)^{i}\theta(u_1,\cdots,\hat{u}_i,\cdots,u_{j-1},\rho(Tu_i)(u_j),u_{j+1},\cdots, u_{k+1}).
               \end{eqnarray}
\emptycomment{It is obvious that for $x\in\g$, we have
$$
\partial_T x(u)=-[x,T(u)]_\g+T(\rho(x)u).
$$ Therefore  $x\in\g$ is closed if and only if
         $
           T\circ\rho(x)=\ad_x\circ T.
         $
              For   $\theta\in  \Hom(V,\g)$, we have
              $$
              \partial_T \theta (u,v)=  [Tu,\theta(v)]_\g+[\theta(u),Tv]_\g-T(\rho(\theta(u))v)-\theta(\rho(Tu)v).
              $$ Therefore, $\theta\in  \Hom(V,\g)$ is closed  if and only if
               $$
              [Tu,\theta(v)]_\g+[\theta(u),Tv]_\g-T(\rho(\theta(u))v)-\theta(\rho(Tu)v)=0,\quad\forall u,v\in V.
               $$}

The coboundary operator  $\partial_T $
  can be alternatively described by the following formula.
\begin{pro}\label{pro:danddT}
 Let $T:V\lon \g$ be an  embedding tensor. Then we have
 $$
 \partial_T \theta=(-1)^{k-1}\Courant{T,\theta},\quad \forall f\in \Hom(\otimes^kV,\g),\,\,k=1,2,\cdots,
 $$
 where the bracket $\Courant{\cdot,\cdot}$ is given by \eqref{eq:glaO}.
\end{pro}
\begin{proof}
It follows from straightforward verification.
\end{proof}

Now we define a cohomology theory governing deformations of an  embedding tensor $T:V\lon\g.$ Define the space of  $0$-cochains $\frkC^0(T)$ to be $0$ and of $1$-cochains $\frkC^1(T)$  to be $\g$. For $n\geq2$, define the space of $n$-cochains $\frkC^n(T)$ as $\frkC^n(T)=\Hom(\otimes^{n-1}V,\g)$.

\begin{defi}\label{de:opcoh}
Let $T$ be an  embedding tensor on  a \LR pair
$((\g,[\cdot,\cdot]_\g),(V;\rho))$. We define the {\bf cohomology of
  the embedding tensor $T$} to be the cohomology of the cochain complex $( \frkC^\bullet(T)=\oplus _{k=0}^{+\infty} \frkC^k(T),\partial_T)$. The corresponding $k$-th cohomology group is denoted by $\huaH^k(T)$.  
\end{defi}

At the end of this section, we study the relation between the
cohomology of an  embedding tensor $T:V\longrightarrow\g$ and
the cohomology of the underlying Leibniz algebra $(V,[\cdot,\cdot]_T)$ given in Proposition \ref{average-Leibniz}.

\begin{thm}
Let  $T:V\longrightarrow \g$ be  an  embedding tensor on  a \LR pair $((\g,[\cdot,\cdot]_\g),(V;\rho))$. Then $\Phi$, defined by  \eqref{eq:phi}, is a homomorphism from the cochain complex   $(\frkC^\bullet(T),\partial_T)$ of the embedding tensor $T$ to the cochain complex $(\frkC^\bullet(V,V),\partial^{reg})$ of the underlying Leibniz algebra $(V,[\cdot,\cdot]_T)$, that is, we have the following commutative diagram:
\begin{equation}
\CD
 \Hom(\otimes^kV,\g) @>\Phi>>      \Hom(\otimes^{k+1}V,V)  \\
  @V \partial_T VV                  @V V \partial^{reg} V                \\
 \Hom(\otimes^{k+1}V,\g) @>\Phi>>               \Hom(\otimes^{k+2}V,V).
\endCD
\label{eq:ext1}
\end{equation}
Consequently, $\Phi$ induces a homomorphism $\Phi_*:\huaH^{k}(T)\longrightarrow H^{k}_{\rm reg}(V,V)$ between the corresponding cohomology groups.
\label{thm:morphismcohomology}
\end{thm}

\begin{proof}
By Proposition \ref{pro:homoglaop} and Proposition \ref{pro:danddT}, for all  $\theta\in\Hom(\otimes^kV,\g)$, we have
\begin{eqnarray*}
 \Phi(\partial_T\theta)=(-1)^{k-1}\Phi\Courant{T,\theta}=(-1)^{k-1}\Courant{\Phi(T),\Phi(\theta)}=\partial^{reg}(\Phi(\theta)).
\end{eqnarray*}
Thus, $\Phi$ is a homomorphism of cochain complexes from $( \frkC^*(T),\partial_T)$ to $(\frkC^*(V,V),\partial^{reg})$ and $\Phi_*$ is a homomorphism between the corresponding  cohomology groups.
\end{proof}

\emptycomment{
\begin{rmk}
If $\rho$ is faithful, then $\Phi$ is injective, realizing
the cochain complex $( \frkC^*(T),\partial_T)$ as a subcomplex of the cochain complex $(\frkC^*(V,V),\partial^{reg})$.
\end{rmk}
}
\section{The controlling $L_\infty$-algebra   of Lie-Leibniz triples}\label{sec:control}

In this section, we apply T. Voronov's higher derived bracket to construct an $L_\infty$-algebra that characterizes Lie-Leibniz triples as Maurer-Cartan elements.

\subsection{$L_\infty$-algebras and higher derived brackets}
\emptycomment{
There is a well know slogan: {\bf every reasonable deformation theory is controlled by an $L_\infty$-algebra (most time a differential graded Lie algebra)}. Surprisingly, every $L_\infty$-algebra can be obtained via higher derived brackets. So, when we want to study deformation theory of mathematical structures we should use higher derived brackets to construct $L_\infty$-algebras which control the deformation theory. Therefore, higher derived brackets become very popular in deformation theory.
}

Let $V^\bullet=\oplus_{k\in\mathbb Z}V^k$ be a $\mathbb Z$-graded vector space.
We will denote by $\Sym(V^\bullet)$ the  symmetric  algebra  of $V^\bullet$. That is,
$
\Sym(V^\bullet):=\Ten(V^\bullet)/I,
$
where $\Ten(V^\bullet)$ is the tensor algebra and $I$ is the $2$-sided ideal of $\Ten(V^\bullet)$ generated by all homogeneous elements  of the form
$
x\otimes y-(-1)^{xy}y\otimes x.
$
We will write $\Sym(V^\bullet)=\oplus_{i=0}^{+\infty}\Sym^i (V)$.
Moreover, we denote the reduced tensor algebra and reduced symmetric  algebra by $\bar{\Ten}V^\bullet:=\oplus_{i=1}^{+\infty}\Ten^{i}(V^\bullet)$ and  $\bar{\Sym}(V^\bullet):=\oplus_{i=1}^{+\infty}\Sym^{i}(V^\bullet)$ respectively. Denote the product of homogeneous elements $v_1,\cdots,v_n\in V^\bullet$ in $\Sym^n(V^\bullet)$ by $v_1\odot\cdots\odot v_n$. The degree of $v_1\odot\cdots\odot v_n$ is by definition the sum of the degrees of $v_i$. For a permutation $\sigma\in\mathbb S_n$ and $v_1,\cdots, v_n\in V^\bullet$,  the Koszul sign $\varepsilon(\sigma)=\varepsilon(\sigma;v_1,\cdots,v_n)\in\{-1,1\}$ is defined by
\begin{eqnarray*}
	v_1\odot\cdots\odot v_n=\varepsilon(\sigma;v_1,\cdots,v_n)v_{\sigma(1)}\odot\cdots\odot v_{\sigma(n)}.
\end{eqnarray*}
The  desuspension operator $s^{-1}$ changes the grading of $V^\bullet$ according to the rule $(s^{-1}V^\bullet)^i:=V^{i+1}$. The  degree $-1$ map $s^{-1}:V^\bullet\lon s^{-1}V^\bullet$ is defined by sending $v\in V^\bullet$ to its   copy $s^{-1}v\in s^{-1}V^\bullet$.

The notion of an $L_\infty$-algebra was introduced by Stasheff in \cite{stasheff:shla}. See  \cite{LS,LM} for more details.
\begin{defi}
An {\bf  $L_\infty$-algebra} is a $\mathbb Z$-graded vector space
$\g^\bullet=\oplus_{k\in\mathbb Z}\g^k$ equipped with a collection
$(k\ge 1)$ of linear maps $l_k:\otimes^k\g^\bullet \lon\g^\bullet$ of degree $1$ with the property that, for any homogeneous elements $x_1,\cdots,x_n\in \g^\bullet$, we have
\begin{itemize}\item[\rm(i)]
{\bf (graded symmetry)}~ for every $\sigma\in\mathbb S_{n}$,
$
l_n(x_{\sigma(1)},\cdots,x_{\sigma(n)})=\varepsilon(\sigma;x_1,\cdots,x_n) l_n(x_1,\cdots,x_n).
$
\item[\rm(ii)] {\bf (generalized Jacobi identity)}~ for all $n\ge 1$,
\begin{eqnarray*}\label{sh-Lie}
\sum_{i=1}^{n}\sum_{\sigma\in \mathbb S_{(i,n-i)} }\varepsilon(\sigma;x_1,\cdots,x_n) l_{n-i+1}(l_i(x_{\sigma(1)},\cdots,x_{\sigma(i)}),x_{\sigma(i+1)},\cdots,x_{\sigma(n)})=0.
\end{eqnarray*}
\end{itemize}
\end{defi}

\begin{rmk}
An $L_\infty$-algebra structure on a graded vector space $\g^\bullet$ is equivalent to a codifferential on  the cofree conilpotent
cocommutative coalgebra $\bar{\Sym}(\g^\bullet)$.
\end{rmk}

\emptycomment{
There is a canonical way to view a differential graded Lie algebra as an $L_\infty$-algebra.

\begin{lem}\label{Quillen-construction}
Let $(\g,[\cdot,\cdot]_\g,d)$ be a dgLa. Then  $(s^{-1}\g,\{l_i\}_{i=1}^{+\infty})$ is an $L_\infty$-algebra, where $
l_1(s^{-1}x)=s^{-1}d(x),~
l_2(s^{-1}x,s^{-1}y)=(-1)^{x}s^{-1}[x,y]_\g,~
l_k=0,$ for all $k\ge 3,
$
and homogeneous elements $x,y\in \g$.
\end{lem}
}

\begin{defi}
The set of {\bf Maurer-Cartan elements} of an $L_\infty$-algebra $(\g^\bullet,\{l_i\}_{i=1}^{+\infty})$ is the set of those $\alpha\in \g^0$ satisfying the Maurer-Cartan equation
$
\sum_{k=1}^{+\infty}\frac{1}{k!}l_k(\alpha,\cdots,\alpha)=0.
$
\end{defi}

Let $\alpha$ be a Maurer-Cartan element. Define $l_k^{\alpha}:\otimes^k\g^\bullet\lon\g^\bullet$  $(k\ge 1)$ by
\begin{eqnarray}
l_k^{\alpha}(x_1,\cdots,x_k)=\sum_{n=0}^{+\infty}\frac{1}{n!}l_{k+n}(\underbrace{\alpha,\cdots,\alpha}_n,x_1,\cdots,x_k).
\end{eqnarray}

\begin{pro}\label{twisted-homotopy-lie}{\rm (\cite{Getzler})}
With the above notation, $(\g^\bullet,\{l_k^{\alpha}\}_{k=1}^{+\infty})$ is an $L_\infty$-algebra. The $L_\infty$-algebra $(\g^\bullet,\{l_k^{\alpha}\}_{k=1}^{+\infty})$  is called the {\bf twisted $L_\infty$-algebra}.
\end{pro}

In the sequel, we recall T. Voronov's derived brackets {\rm \cite{Vo}}, which is a useful tool to construct  explicit  $L_\infty$-algebras.

\begin{defi}
 $V$-data consist of a quadruple $(L,\h,P,\Delta)$ where
\begin{itemize}
\item[$\bullet$] $(L,[\cdot,\cdot])$ is a graded Lie algebra,
\item[$\bullet$] $\h$ is an abelian graded Lie subalgebra of $(L,[\cdot,\cdot])$,
\item[$\bullet$] $P:L\lon L$ is a projection, that is $P\circ P=P$, whose image is $\h$ and kernel is a  graded Lie subalgebra of $(L,[\cdot,\cdot])$,
\item[$\bullet$] $\Delta$ is an element of $ \ker(P)^1$ such that $[\Delta,\Delta]=0$.
\end{itemize}
\end{defi}

\emptycomment{
\begin{thm}{\rm (\cite{Vo})}\label{thm:db}
Let $(L,\h,P,\Delta)$ be a $V$-data. Then $(\h,\{{l_k}\}_{k=1}^{+\infty})$ is an $L_\infty$-algebra where
\begin{eqnarray}\label{V-shla}
l_k(a_1,\cdots,a_k)=P\underbrace{[\cdots[[}_k\Delta,a_1],a_2],\cdots,a_k],
\end{eqnarray}
 for homogeneous $ a_1,\cdots,a_k\in\h$. We call $\{{l_k}\}_{k=1}^{+\infty}$ the {\bf higher derived brackets} of the $V$-data $(L,\h,P,\Delta)$.
\end{thm}
}

\begin{thm}{\rm (\cite{Vo})}\label{thm:db-big-homotopy-lie-algebra}
Let $(L,\h,P,\Delta)$ be  $V$-data. Then the graded vector space $s^{-1}L\oplus \h$  is an $L_\infty$-algebra where
\begin{eqnarray*}\label{V-shla-big-algebra}
l_1(s^{-1}x,a)&=&(-s^{-1}[\Delta,x],P(x+[\Delta,a])),\\
l_2(s^{-1}x,s^{-1}y)&=&(-1)^xs^{-1}[x,y],\\
l_k(s^{-1}x,a_1,\cdots,a_{k-1})&=&P[\cdots[[x,a_1],a_2],\cdots,a_{k-1}],\quad k\geq 2,\\
l_k(a_1,\cdots,a_{k-1},a_k)&=&P[\cdots[[\Delta,a_1],a_2],\cdots,a_{k}],\quad k\geq 2.
\end{eqnarray*}
Here $a,a_1,\cdots,a_k$ are homogeneous elements of $\h$ and $x,y$ are homogeneous elements of $L$. All other $L_\infty$-algebra products that are not obtained from the ones written above by permutations of arguments, will vanish.  

Let $L'$ be a graded Lie subalgebra of $L$ that satisfies $[\Delta,L']\subset L'$. Then $s^{-1}L'\oplus \h$ is an $L_\infty$-subalgebra of the above $L_\infty$-algebra. In particular, $(\h,\{{l_k}\}_{k=1}^{+\infty})$ is an $L_\infty$-algebra, where
\begin{eqnarray}\label{V-shla}
l_k(a_1,\cdots,a_k)=P\underbrace{[\cdots[[}_k\Delta,a_1],a_2],\cdots,a_k], \quad  \mbox{for }~   a_1,\cdots,a_k\in\h.
\end{eqnarray}
\end{thm}

  $L_\infty$-algebras were constructed using the above method to study simultaneous deformations of morphisms between Lie algebras in \cite{Fregier-Zambon-1,Fregier-Zambon-2}, and to study simultaneous deformations of relative Rota-Baxter Lie algebras in \cite{LST}.

\subsection{The controlling $L_\infty$-algebra of Lie-Leibniz triples}

Let $\g$ and $V$ be two vector spaces. Then we have a graded Lie
algebra $(\oplus_{n=0}^{+\infty}C^{n}(\g\oplus V,\g\oplus
V),[\cdot,\cdot]_{\B})$. This graded Lie algebra gives rise to
V-data therefore also to an $L_\infty$-algebra.
\begin{pro}\label{pro:VdataL}
We have $V$-data $(L,\h,P,\Delta)$ as follows:
\begin{itemize}
\item[$\bullet$] the graded Lie algebra $(L,[\cdot,\cdot])$ is given by $\big(\oplus_{n=0}^{+\infty}C^{n}(\g\oplus V,\g\oplus V),[\cdot,\cdot]_{\B}\big)$;
\item[$\bullet$] the abelian graded Lie subalgebra $\h$ is given by
\begin{equation}\label{defi:h}
\h:=\oplus_{n=0}^{+\infty}\Hom(\otimes^{n+1}V,\g);
\end{equation}
\item[$\bullet$] $P:L\lon L$ is the projection onto the subspace $\h$;
\item[$\bullet$] $\Delta=0$.
\end{itemize}

Consequently, we obtain an $L_\infty$-algebra $(s^{-1}L\oplus\h,\{l_k\}_{k=1}^{+\infty})$, where $l_i$ are given by
\begin{eqnarray*}
l_1(s^{-1}Q,\theta)&=&P(Q),\\
  l_2(s^{-1}Q,s^{-1}Q')&=&(-1)^Qs^{-1}[Q,Q']_{\B}, \\
l_k(s^{-1}Q,\theta_1,\cdots,\theta_{k-1})&=&P[\cdots[Q,\theta_1]_{\B},\cdots,\theta_{k-1}]_{\B},
\end{eqnarray*}
for homogeneous elements   $\theta,\theta_1,\cdots,\theta_{k-1}\in \h$, homogeneous elements  $Q,Q'\in L$ and all the other possible combinations vanish.
\end{pro}
\begin{proof}
  Note that
  $
  \h=\oplus_{n=0}^{+\infty}\Hom(\otimes^{n+1}V,\g)
  $
  is an abelian subalgebra of $(L,[\cdot,\cdot])$. Since $P$ is the projection onto $\h$, it is obvious that $P\circ P=P$. It is also straightforward to see that the kernel  of $P$ is a  graded Lie subalgebra of $(L,[\cdot,\cdot])$. Thus $(L,\h,P,\Delta=0)$ are V-data.

   The other conclusions follows immediately from Theorem \ref{thm:db-big-homotopy-lie-algebra}.
\end{proof}

By Theorem \ref{lem:MC-algrep}, we obtain that
\begin{equation}\label{defi:Lprime}
L'=\oplus_{n=0}^{+\infty}\hat{C}^{n}(\g\oplus V,\g\oplus V),\quad\mbox{where}\quad \hat{C}^{n}(\g\oplus V,\g\oplus V)=\Hom(\wedge^{n+1}\g,\g)\oplus\Hom(\wedge^{n}\g\otimes V,V)
 \end{equation}is a graded Lie subalgebra of $\big(\oplus_{n=0}^{+\infty}C^{n}(\g\oplus V,\g\oplus V),[\cdot,\cdot]_{\B}\big)$. Then we have the following result.

\begin{pro}\label{cor:Linfty}
 With above notations,  $(s^{-1}L'\oplus \h,\{l_i\}_{i=1}^{+\infty})$ is an $L_\infty $-algebra, where $l_i$ are given by
 \begin{eqnarray*}
  l_2(s^{-1}Q,s^{-1}Q')&=&(-1)^Qs^{-1}[Q,Q']_{\B}, \\
l_k(s^{-1}Q,\theta_1,\cdots,\theta_{k-1})&=&P[\cdots[Q,\theta_1]_{\B},\cdots,\theta_{k-1}]_{\B},
\end{eqnarray*}
for homogeneous elements   $\theta_1,\cdots,\theta_{k-1}\in \h$, homogeneous elements  $Q,Q'\in L'$, and all the other possible combinations vanish.
\end{pro}
\begin{proof}
  It follows from Theorem \ref{thm:db-big-homotopy-lie-algebra} and Proposition \ref{pro:VdataL}.
\end{proof}

Now we are ready to give the controlling $L_\infty$-algebra of Lie-Leibniz triples, which is the main result in this subsection.

\begin{thm}\label{deformation-rota-baxter}
  Let $\g$ and $V$ be two vector spaces,  $\mu\in\Hom(\wedge^2\g,\g),~\rho\in\Hom(\g\otimes V,V)$ and $T\in\Hom(V,\g)$. Then $((\g,\mu),(V;\rho),T)$ is a Lie-Leibniz triple if and only if  $(s^{-1}(\mu\boxplus\rho),T)$ is a Maurer-Cartan element of the   $L_\infty$-algebra $(s^{-1}L'\oplus \h,\{l_i\}_{i=1}^{+\infty})$ given in Proposition \ref{cor:Linfty}.
\end{thm}

\begin{proof}
 Let $(s^{-1}(\mu\boxplus\rho),T)$ be a Maurer-Cartan element of $(s^{-1}L'\oplus \h,\{l_i\}_{i=1}^{+\infty})$. Then we have
\begin{eqnarray*}
&&\sum_{k=1}^{+\infty}\frac{1}{k!}l_k\Big((s^{-1}(\mu\boxplus\rho),T),\cdots,(s^{-1}(\mu\boxplus\rho),T)\Big)\\
&=&\frac{1}{2!}l_2\Big((s^{-1}(\mu\boxplus\rho),T),(s^{-1}(\mu\boxplus\rho),T)\Big)+\frac{1}{3!}l_3\Big((s^{-1}(\mu\boxplus\rho),T),(s^{-1}(\mu\boxplus\rho),T),(s^{-1}(\mu\boxplus\rho),T)\Big)\\
&=&\Big(-s^{-1}\frac{1}{2}[\mu\boxplus\rho,\mu\boxplus\rho]_{\B},\frac{1}{2}[[\mu\boxplus\rho,T]_{\B},T]_{\B}\Big)\\
&=&(0,0).
\end{eqnarray*}
Thus, we obtain $$
 [\mu\boxplus\rho,\mu\boxplus\rho]_{\B}=0,\quad
[[\mu\boxplus\rho,T]_{\B},T]_{\B}=0.
$$
 By Theorem \ref{lem:MC-algrep} and Theorem \ref{average-MC},  $(\g,\mu)$ is a Lie algebra, $(V;\rho)$ is its representation and   $T$ is an embedding tensor on  the \LR pair $((\g,\mu),(V;\rho))$.
\end{proof}

\section{Cohomology of Lie-Leibniz triples and applications}\label{sec:cohomology}

In this section, we introduce a cohomology theory of Lie-Leibniz
triples and justify it by using it to classify infinitesimal
deformations and central extensions.
\subsection{Regular cohomology of Lie-Leibniz triples   and infinitesimal deformations}\label{sec:cohomology1}

In this subsection,  first we recall the cohomology of a \LR pair, and then we introduce a regular cohomology of a Lie-Leibniz triple. Finally, we use the second cohomology group to characterize infinitesimal deformations of a Lie-Leibniz triple.

Let $((\g,\mu),(V;\rho))$ be a \LR pair.  Define the set of $0$-cochains $\frkC^0(\g,\rho)$ to be $0$. For $n\geq 1$, we define the set of $n$-cochains $\frkC^n(\g,\rho)$  by
$$
\frkC^n(\g,\rho)=\Hom(\wedge^{n}\g,\g)\oplus \Hom(\wedge^{n-1}\g\otimes V,V).
$$
Define the coboundary operator $\delta:\frkC^n(\g,\rho)\lon \frkC^{n+1}(\g,\rho)$ by
\begin{equation}\label{eq:defdel}
\delta f=(-1)^{n-1}[ \mu\boxplus\rho, {f}]_{\B}.
\end{equation}
By Theorem \ref{lem:MC-algrep}, $\delta^2=0$. Thus we obtain a cochain complex $(\oplus_n\frkC^n(\g,\rho),\delta)$.
\begin{defi}
  The cohomology of the cochain complex $(\oplus_n\frkC^n(\g,\rho),\delta)$ is called the {\bf cohomology of the \LR pair}.
\end{defi}

Now we give the precise formula of the coboundary operator $\delta$. Write $f=(f_\g,f_V)$ and $\delta f=\Big((\delta f)_\g,(\delta f)_V\Big)$, where $f_\g\in \Hom(\wedge^n\g,\g)$, $f_V\in\Hom(\wedge^{n-1}\g\otimes V,V)$, $(\delta f)_\g\in \Hom(\wedge^{n+1}\g,\g)$ and $(\delta f)_V\in\Hom(\wedge^{n}\g\otimes V,V)$.  Then we have
$$
(\delta f)_\g=(-1)^{n-1}[ {\mu} , {f}_\g]_{\B}=\dM_\CE f_\g,
$$
where $ \dM_\CE: \Hom(\wedge^n\g,\g)\lon  \Hom(\wedge^{n+1}\g,\g)$ is the  {\bf Chevalley-Eilenberg coboundary operator} of the Lie algebra $(\g,[\cdot,\cdot]_\g)$, and $(\delta f)_V$  is given by
\begin{eqnarray*}
\nonumber&&(\delta f)_V(x_1,\cdots,x_{n},v)=(-1)^{n-1}[ \mu\boxplus\rho , f]_{\B}(x_1,\cdots,x_{n},v)\\
&=&(-1)^{n-1}\Big((\mu\boxplus\rho)\bar{\circ} f-(-1)^{n-1}f\bar{\circ}(\mu\boxplus\rho)\Big)(x_1,\cdots,x_{n},v)\\
\label{cohomology-algebra-rep-V}&=&\sum_{1\le i<j\le n}(-1)^if_V(x_1,\cdots,\hat{x_i},\cdots,x_{j-1},[x_i,x_j]_\g,x_{j+1},\cdots,x_{n},v)+(-1)^{n-1}\rho(f_\g(x_1,\cdots,x_{n}))v\\
\nonumber&&+\sum_{i=1}^{n}(-1)^{i+1}\Big(\rho(x_i)f_V(x_1,\cdots,\hat{x}_i,\cdots,x_n,v)-f_V\big(x_1,\cdots,\hat{x}_i,\cdots,x_n,\rho(x_i)v\big)\Big),
\end{eqnarray*}
for all $x_1,\cdots,x_{n}\in\g$ and $v\in V.$

Now we are ready to define the cohomology of a Lie-Leibniz triple. Let $((\g,\mu),(V;\rho),T)$ be a Lie-Leibniz triple, i.e. $\rho:\g\lon\gl(V)$ is a representation of the Lie algebra $(\g,\mu)$ and $T:V\lon\g$ is an embedding tensor. Define the set of  $0$-cochains $\frkC^0(\g,\rho,T)$ to be $0$. For $n\geq 1$, define the space of    $n$-cochains $\frkC^n(\g,\rho,T)$ by
\begin{eqnarray*}
\frkC^n(\g,\rho,T)&:=&\frkC^n(\g,\rho)\oplus \frkC^n(T)\\
&=&\Big(\Hom(\wedge^n\g,\g)\oplus \Hom(\wedge^{n-1}\g\otimes V,V)\Big)\oplus\Hom(\otimes^{n-1}V,\g).
\end{eqnarray*}

Define the {\bf coboundary operator} $\huaD:\frkC^n(\g,\rho,T)\lon \frkC^{n+1}(\g,\rho,T)$ by
\begin{eqnarray}\label{cohomology-of-RB}
  \huaD(f,\theta)&=&(-1)^{n-2}\big(-[\mu\boxplus\rho,f]_{\B},[[\mu\boxplus\rho,T]_{\B},\theta]_{\B}+\frac{1}{n!}\underbrace{[\cdots[[}_nf,T]_{\B},T]_{\B},\cdots,T]_{\B}\big)\\
\label{eq:Dcomp} &=&(\delta f, \partial_T \theta+\Omega_Tf), \quad \forall f\in \frkC^n(\g,\rho), \theta\in \frkC^n(T),
\end{eqnarray}
where $\delta$ and $\partial_T$ are given by \eqref{eq:defdel} and \eqref{eq:defpt}, and $\Omega_T:\frkC^n(\g,\rho)\lon \frkC^{n+1}(T)$ is defined by
\begin{eqnarray}\label{key-cohomology-T-abstract}
\Omega_Tf:=(-1)^{n-2} \frac{1}{n!}\underbrace{[\cdots[[}_nf,T]_{\B},T]_{\B},\cdots,T]_{\B}.
\end{eqnarray}
The precise formula for $\Omega_T$ is given as follows.
\begin{lem}The operator $\Omega_T:\Hom(\wedge^n\g,\g)\oplus \Hom(\wedge^{n-1}\g\otimes V,V)\lon\Hom(\otimes^{n}V,\g)$ is given by
\begin{eqnarray}\label{key-cohomology-T}
 \Omega_T(f_\g,f_V)(v_1,\cdots,v_n)=(-1)^{n}\Big(f_\g(Tv_1,\cdots,Tv_n)-Tf_V (Tv_1,\cdots,Tv_{n-1}, v_n)\Big),
\end{eqnarray}
where  $f_\g\in \Hom(\wedge^n\g,\g),~f_V\in \Hom(\wedge^{n-1}\g\otimes V,V)$ and $v_1,\cdots,v_n\in V.$

\end{lem}

\begin{proof}
  By Remark \ref{B-coder}, it is convenient to view the elements of $\oplus_{n=0}^{+\infty}C^{n}(\g\oplus V;\g\oplus V)$ as coderivations of $\bar{\Ten}(\g\oplus V)$. The coderivations corresponding to $f$ and $T$ will be denoted by $\bar{f}$ and $\bar{T}$ respectively.
 Then, by induction, we have
\begin{eqnarray*}
&&\underbrace{[\cdots[[}_nf,T]_{\B},T]_{\B},\cdots,T]_{\B}\big((x_1,v_1),\cdots,(x_{n},v_{n})\big)\\
&=&\sum_{i=0}^{n}(-1)^{i}{n\choose i}\big(\underbrace{\bar{T}\circ\cdots\circ\bar{T}}_i\circ (\bar{f}_\g+\bar{f}_V)\circ \underbrace{\bar{T}\cdots\circ \bar{T}}_{n-i}\big)\big((x_1,v_1),\cdots,(x_{n},v_{n})\big)\\
&=&\big(n!f_\g(Tv_1,\cdots,Tv_n),0\big)+\Big((-1)^1n(n-1)!Tf_V\big(Tv_1,\cdots,Tv_{n-1},v_n\big),0\Big)\\
&=&n!\Big (f_\g(Tv_1,\cdots,Tv_n)-Tf_V\big(Tv_1,\cdots,Tv_{n-1},v_n\big),0\Big),
\end{eqnarray*}
which implies that \eqref{key-cohomology-T} holds.
\end{proof}

\begin{thm}\label{cohomology-of-average}
  With the above notations,  $(\oplus _{n=0}^{+\infty}\frkC^n(\g,\rho,T),\huaD)$ is a cochain complex, i.e. $\huaD\circ \huaD=0.$
\end{thm}
\begin{proof} By Theorem \ref{deformation-rota-baxter},
      $(s^{-1}(\mu\boxplus\rho),T)$ is a Maurer-Cartan element of the   $L_\infty$-algebra $(s^{-1}L'\oplus \h,\{l_i\}_{i=1}^{+\infty})$ given in Proposition \ref{cor:Linfty}. By Proposition \ref{twisted-homotopy-lie}, there is a twisted  $L_\infty$-algebra $(s^{-1}L'\oplus \h,\{l_i^{(s^{-1}(\mu\boxplus\rho),T)}\}_{i=1}^{+\infty})$.
For any $(f,\theta)\in\frkC^n(\g,\rho,T)$, we have $(s^{-1}f,\theta)\in (s^{-1}L'\oplus \h)^{n-2}$ and
\begin{eqnarray*}
  l_1^{(s^{-1}(\mu\boxplus\rho),T)}(s^{-1}f,\theta)&=&\sum_{n=0}^{+\infty}\frac{1}{n!}l_{k+n}\Big(\underbrace{(s^{-1}(\mu\boxplus\rho),T),\cdots,(s^{-1}(\mu\boxplus\rho),T)}_n,(s^{-1}f,\theta)\Big)\\
 &=&l_2(s^{-1}(\mu\boxplus\rho),s^{-1}f)+l_3(s^{-1}(\mu\boxplus\rho),T,\theta)+\frac{1}{n!}l_{n+1}(f,\underbrace{T,\cdots,T}_n)\\
                        &=&\big(-s^{-1}[\mu\boxplus\rho,f]_{\B},[[\mu\boxplus\rho,T]_{\B},\theta]_{\B}+\frac{1}{n!}\underbrace{[\cdots[[}_nf,T]_{\B},T]_{\B},\cdots,T]_{\B}\big).
\end{eqnarray*}
By \eqref{cohomology-of-RB}, we deduce that
$$
\huaD(f,\theta)=(-1)^{n-2} l_1^{(s^{-1}(\mu+\rho),T)}(s^{-1}f,\theta).
$$
By  $ l_1^{(s^{-1}(\mu+\rho),T)}\circ  l_1^{(s^{-1}(\mu+\rho),T)}=0$, we obtain that $(\oplus _{n=0}^{+\infty}\frkC^n(\g,\rho,T),\huaD)$ is a cochain complex.
\end{proof}

About the relation between the operator $\delta$, $\partial_T$ and $\Omega_T$, we have
\begin{cor}\label{cor:relation}
With the above notations, we have $\Omega_T\circ \delta+\partial_T\circ \Omega_T=0$.
\end{cor}
 \begin{proof}
   For all $(f,\theta) \in \frkC^n(\g,\rho,T)$, by the fact $\delta\circ \delta=\partial_T\circ\partial_T=0,$ we have
   $$
   0=(\huaD\circ \huaD)(f,\theta)=\huaD(\delta f, \partial_T \theta+\Omega_Tf)=(\delta(\delta f),\partial_T(\partial_T \theta+\Omega_Tf)+\Omega_T(\delta f)),
   $$
   which implies that $\Omega_T\circ \delta+\partial_T\circ \Omega_T=0$.
 \end{proof}

\begin{defi}
  The cohomology of the cochain complex $(\oplus _{n=0}^{+\infty}\frkC^n(\g,\rho,T),\huaD)$ is called the {\bf regular cohomology  of the Lie-Leibniz triple} $((\g,\mu),(V;\rho),T)$. We denote its $n$-th cohomology group by $\huaH^n_{\rm reg}(\g,\rho,T)$
\end{defi}

The formula of the coboundary operator $\huaD$ can be well-explained by the following diagram:
 \[
\small{ \xymatrix{
\cdots
\longrightarrow \frkC^n(\g,\rho)\ar[dr]^{\Omega_T} \ar[r]^{\qquad\delta} & \frkC^{n+1}(\g,\rho) \ar[dr]^{\Omega_T} \ar[r]^{\delta\qquad}  & \frkC^{n+2}(\g,\rho)\longrightarrow\cdots  \\
\cdots\longrightarrow \frkC^n(T) \ar[r]^{\qquad\partial_T} &\frkC^{n+1}(T)\ar[r]^{\partial_T\qquad}&\frkC^{n+2}(T)\longrightarrow \cdots.}
}
\]
\begin{thm}\label{cohomology-exact}
Let $((\g,\mu),(V;\rho),T)$ be a Lie-Leibniz triple. Then there is a short exact sequence of   cochain complexes:
$$
0\longrightarrow(\oplus_{n=0}^{+\infty}\frkC^n(T),\partial_T)\stackrel{\iota}{\longrightarrow}(\oplus _{n=0}^{+\infty}\frkC^n(\g,\rho,T),\huaD)\stackrel{p}{\longrightarrow} (\oplus_{n=0}^{+\infty}\frkC^n(\g,\rho),\delta)\longrightarrow 0,
$$
where $\iota$ and $p$ are the inclusion map and the projection map.

Consequently, there is a long exact sequence of the  cohomology groups:
$$
\cdots\longrightarrow\huaH^n(T)\stackrel{\huaH^n(\iota)}{\longrightarrow}\huaH^n_{\rm reg}(\g,\rho,T)\stackrel{\huaH^n(p)}{\longrightarrow} \huaH^n(\g,\rho)\stackrel{c^n}\longrightarrow \huaH^{n+1}(T)\longrightarrow\cdots,
$$
where the connecting map $c^n$ is defined by
$
c^n([f])=[\Omega_Tf],$  for all $[f]\in \huaH^n(\g,\rho).$
\end{thm}
\begin{proof}
 By  \eqref{eq:Dcomp}, we have the short exact sequence  of cochain complexes which induces a long exact sequence of cohomology groups.   Also by \eqref{eq:Dcomp},   $c^n$ is given by
$
c^n([f])=[\Omega_Tf].$
\end{proof}

At the end of this section, we study infinitesimal deformations of Lie-Leibniz triples.
Let $\K[[t]]$ be the ring of formal power series in one variable
$t$. Let $((\g,\mu),(V;\rho),T)$ be a Lie-Leibniz triple. Let $\g[[t]]$ and $V[[t]]$ be the spaces of  formal power series in $t$ with coefficients in $\g$ and $V$ respectively.
\begin{defi}
If $[\cdot,\cdot]_t=[\cdot,\cdot]_\g+t\omega$ defines a Lie algebra structure on $\g[[t]]/(t^{2})$, $\rho_t=\rho+t\varrho$ defines a representation of the Lie algebra $(\g[[t]]/(t^{2}),[\cdot,\cdot]_t)$ on $ V[[t]]/(t^{2})$ and $T_t=T+t\huaT:V[[t]]/(t^{2})\lon\g[[t]]/(t^{2})$  satisfies
\begin{eqnarray}\label{A-operator1}
[T_tu,T_tv]_t=T_t\Big(\rho_t(T_tu)(v)\Big),\;\;\forall
u,v\in V,
\end{eqnarray} for $\omega\in\Hom(\wedge^2\g,\g),~\varrho\in\Hom(\g\otimes V,V)$ and $\huaT:V\lon\g$, we say that $(\omega,\varrho,\huaT)$ generates
  an {\bf  infinitesimal deformation}  of the Lie-Leibniz triple $((\g,\mu),(V;\rho),T)$.
\end{defi}

Let $(\omega,\varrho,\huaT)$ generate an   infinitesimal deformation. By the fact that $[\cdot,\cdot]_t=[\cdot,\cdot]_\g+t\omega$ defines a Lie algebra structure on $\g[[t]]/(t^{2})$, we get
\begin{equation}\label{eq:alg1}
  \dM_\CE\omega=0.
\end{equation}
   Then since $(V[[t]]/(t^{2});\rho_t)$ is a representation of $(\g[[t]]/(t^{2}),[\cdot,\cdot]_t)$, we obtain
   \begin{eqnarray}
    \label{eq:rep1} \rho(\omega(x,y))+\varrho([x,y]_\g)&=&[\rho(x),\varrho(y)]+[\varrho(x),\rho(y)].
   \end{eqnarray}
By \eqref{A-operator1}, we deduce that
   \begin{eqnarray}
    \label{eq:ro1} [\huaT u,Tv]_\g+[Tu,\huaT v]_\g+\omega(Tu,Tv)&=&T\Big(\rho(\huaT u)v+\varrho(Tu)v\Big)+\huaT\big(\rho(Tu)v\big).
   \end{eqnarray}

       \begin{pro}
   The triple $(\omega,\varrho,\huaT)$ generates an  infinitesimal  deformation    if and only if  $(\omega,\varrho,\huaT)$ is a $2$-cocycle of   the Lie-Leibniz triple $((\g,[\cdot,\cdot]_\g),(V;\rho),T)$.
    \end{pro}
    \begin{proof}
      By \eqref{eq:alg1}, \eqref{eq:rep1} and \eqref{eq:ro1},  we deduce that  $(\omega,\varrho,\huaT)$ generates an  infinitesimal  deformation  of the Lie-Leibniz triple $((\g,[\cdot,\cdot]_\g),(V;\rho),T)$ if and only if    $(\omega,\varrho,\huaT)$ is a $2$-cocycle.
    \end{proof}

    In the sequel, we define equivalences between infinitesimal deformations of a Lie-Leibniz triple and show that infinitesimal deformations of a Lie-Leibniz triple are classified by its second cohomology group.

     \begin{defi}
    Two infinitesimal deformations of a Lie-Leibniz triple  $((\g,[\cdot,\cdot]_\g),(V;\rho),T)$ generated by $(\omega,\varrho,\huaT)$ and $(\omega',\varrho',\huaT')$ are said to be {\bf equivalent} if there exist $N\in\gl(\g)$,  $S\in\gl(V)$ and $x\in\g$ such that $(\Id_\g+tN+t\ad_x,\Id_V+tS+t\rho(x))$ is a homomorphism from $((\g[[t]]/(t^{2}),[\cdot,\cdot]_\g+t\omega'),(V[[t]]/(t^{2});\rho+t\varrho'),T+t\huaT')$ to $((\g[[t]]/(t^{2}),[\cdot,\cdot]_\g+t\omega),(V[[t]]/(t^{2}),\rho+t\varrho),T+t\huaT)$.
    \end{defi}
    Since $\Id_\g+tN+t\ad_x$ is a Lie algebra morphism from $(\g[[t]]/(t^{2}),[\cdot,\cdot]_\g+t\omega')$ to $(\g[[t]]/(t^{2}),[\cdot,\cdot]_\g+t\omega)$, we get
    \begin{equation}
      \label{eq:equmor1} \omega'-\omega=\dM_\CE N.
    \end{equation}
  By the equality $(\Id_V+tS+t\rho(x))(\rho+t\varrho')(y)u=(\rho+t\varrho)\big((\Id_\g+tN+t\ad_x)y)(\Id_V+tS+t\rho(x))u$, we deduce that
     \begin{equation}
      \label{eq:equmor2} \varrho'(y)u- \varrho(y)u= \rho(Ny)u+ \rho(y)Su- S\rho(y)u,\quad \forall y\in\g, u\in V.
    \end{equation}
    By the equality $(\Id_\g+tN+t\ad_x)\circ (T+t\huaT')=(T+t\huaT)\circ (\Id_V+tS+t\rho(x))$, we obtain
    \begin{equation}
     \label{eq:equmor3}  \huaT'-\huaT=T\circ \rho(x)-\ad_x\circ T-N\circ T+T\circ S.
    \end{equation}

    \begin{thm}
       Let $((\g,[\cdot,\cdot]_\g),\rho,T)$ be a Lie-Leibniz triple. If two infinitesimal  deformations   generated by $(\omega,\varrho,\huaT)$ and $(\omega',\varrho',\huaT')$ are   equivalent, then $(\omega,\varrho,\huaT)$ and $(\omega',\varrho',\huaT')$ are in the same cohomology class in $\huaH^2_{\rm reg}(\g,\rho,T)$.
    \end{thm}
   \begin{proof} By \eqref{eq:equmor1}, \eqref{eq:equmor2} and \eqref{eq:equmor3}, we deduce that
    $$
    (\omega',\varrho',\huaT')-(\omega,\varrho,\huaT)=\huaD(N,S,x),
    $$
    which implies that $(\omega,\varrho,\huaT)$ and $(\omega',\varrho',\huaT')$ are in the same cohomology class.
    \end{proof}

\subsection{Cohomology  with arbitrary coefficients and central extensions}\label{sec:cohomology2}
In this subsection,    we introduce the cohomology of a Lie-Leibniz triple with coefficients in an arbitrary representation and classify  central extensions of a Lie-Leibniz triple using the second cohomology group.

With the help of the Lie-Leibniz triple given in Example \ref{example-3}, we give the notion of a representation of a Lie-Leibniz triple as follows.
\begin{defi} \label{defi:RepLieLei}
  A {\bf representation of a Lie-Leibniz triple}   $((\g,[\cdot,\cdot]_\g),(V;\rho),T)$ on a 2-term complex of vector spaces $W\stackrel{\frkT}{\lon}\h$ is a Lie-Leibniz triple  homomorphism $(\phi,\varphi)$ from   $(\g,V,T)$ to $(\End (W\stackrel{\frkT}{\lon}\h),\Hom(\h,W),\overline{\frkT})$. More precisely, it consists of a Lie algebra homomorphism
    $\phi:\g\longrightarrow\End (W\stackrel{\frkT}{\lon}\h)$ and a linear map $\varphi:V\longrightarrow \Hom(\h,W)$ such that
         \begin{eqnarray}
          \overline{\frkT}\circ \varphi&=&\phi\circ T,\label{defi:repcon1}\\
                \varphi\rho(x)(u)&=&[\phi(x),\varphi(u)]_C,\quad\forall x\in\g, u\in V.\label{defi:repcon2}
      \end{eqnarray}
\end{defi}

Usually we will denote a representation by $(W\stackrel{\frkT}{\lon}\h,\phi,\varphi).$ Since $\End (W\stackrel{\frkT}{\lon}\h)\subset \gl(\h)\oplus \gl(W)$, for any $x\in\g$, we will always write $\phi(x)=(\phi_\h(x),\phi_W(x))$ for $\phi_\h(x)\in\gl(\h)$ and $\phi_W(x)\in\gl(W)$. The following result follows from straightforward verifications.

\begin{pro}\label{pro:semiLieLeibniz}
  Let $((\g,[\cdot,\cdot]_\g),(V;\rho),T)$ be a Lie-Leibniz triple  and  $(W\stackrel{\frkT}{\lon}\h,\phi,\varphi)$ its representation. Then $((\g\oplus \h,[\cdot,\cdot]_{\phi_\h}), (V\oplus W,\rho+\phi_W+\varphi),T+\frkT)$ is a Lie-Leibniz triple, where $[\cdot,\cdot]_{\phi_\h}$ is the semidirect product Lie bracket given by
  \begin{eqnarray}\label{rep-1s}
  [x+\alpha,y+\beta]_{\phi_\h}=[x,y]_\g+{\phi_\h}(x)\beta-{\phi_\h}(y)\alpha,\quad\forall x,y\in\g, \alpha,\beta\in\h,
  \end{eqnarray}
  and the representation $\rho+\phi_W+\varphi$ of the Lie algebra $(\g\oplus \h,[\cdot,\cdot]_{\phi_\h})$ on $V\oplus W$ is given by
 \begin{eqnarray}\label{rep-2s}
\qquad  (\rho+\phi_W+\varphi)(x+\alpha)(u+\xi)=\rho(x)u+\phi_W(x)\xi-\varphi(u)\alpha,\quad\forall x\in\g, \alpha\in\h, u\in V, \xi\in W.
  \end{eqnarray}
\end{pro}

This Lie-Leibniz triple is called the {\bf semidirect product} of $((\g,[\cdot,\cdot]_\g),(V;\rho),T)$ and the representation $(W\stackrel{\frkT}{\lon}\h,\phi,\varphi)$.

Let $(W\stackrel{\frkT}{\lon}\h,\phi,\varphi)$ be a representation of  a Lie-Leibniz triple    $((\g,[\cdot,\cdot]_\g),(V;\rho),T)$. Define the set of  $0$-cochains $\frkC^0(\g,\rho,T,\phi,\varphi,\frkT)$ to be $0$. For $n\geq 1$, define the space of     $n$-cochains $\frkC^n(\g,\rho,T,\phi,\varphi,\frkT)$ by
\begin{eqnarray*}
\frkC^n(\g,\rho,T,\phi,\varphi,\frkT)
&=&\Big(\Hom(\wedge^n\g,\h)\oplus \Hom(\wedge^{n-1}\g\otimes V,W)\Big)\oplus\Hom(\otimes^{n-1}V,\h).
\end{eqnarray*}

Define the {\bf coboundary operator} $\huaD_R:\frkC^n(\g,\rho,T,\phi,\varphi,\frkT)\lon \frkC^{n+1}(\g,\rho,T,\phi,\varphi,\frkT)$ by
\begin{eqnarray*}
  \huaD_R(f,\theta)=(\delta f, \partial \theta+\Omega f),
\end{eqnarray*}
for  all $f=(f_\g,f_V)\in \Hom(\wedge^n\g,\h)\oplus \Hom(\wedge^{n-1}\g\otimes V,W),~ \theta\in \Hom(\otimes^{n-1}V,\h)$. Here $\delta,\partial$ and $\Omega$ are given as follows:
\begin{itemize}
\item Write $\delta f=((\delta f)_\g,(\delta f)_V)\in \Hom(\wedge^{n+1}\g,\h)\oplus \Hom(\wedge^{n}\g\otimes V,W)$. Then $(\delta f)_\g=\dM_\CE f_\g$, where $\dM_\CE:\Hom(\wedge^n\g,\h)\lon \Hom(\wedge^{n+1}\g,\h)$ is the
  Chevalley-Eilenberg coboundary operator  of the Lie algebra $(\g,[\cdot,\cdot]_\g)$ with coefficients in $(\h,{\phi_\h})$, and $(\delta f)_V$  is given by
\begin{eqnarray*}
\nonumber&&(\delta f)_V(x_1,\cdots,x_{n},v)\\
&=&\sum_{1\le i<j\le n}(-1)^if_V(x_1,\cdots,\hat{x_i},\cdots,x_{j-1},[x_i,x_j]_\g,x_{j+1},\cdots,x_{n},v)+(-1)^{n}\varphi(v)f_\g(x_1,\cdots,x_{n}) \\
\nonumber&&+\sum_{i=1}^{n}(-1)^{i+1}\Big(\phi_W(x_i)f_V(x_1,\cdots,\hat{x}_i,\cdots,x_n,v)-f_V\big(x_1,\cdots,\hat{x}_i,\cdots,x_n,\rho(x_i)v\big)\Big),
\end{eqnarray*}
for all $x_1,\cdots,x_{n}\in\g$ and $v\in V.$

\item $\partial: \Hom(\otimes^{n-1}V,\h)\longrightarrow   \Hom(\otimes^{n}V,\h)$ is given by
               \begin{eqnarray*}
              \partial \theta(u_1,\cdots,u_{n})&=&\sum_{i=1}^{n}(-1)^{i+1}\phi_\h(Tu_i)\theta(u_1,\cdots,\hat{u_i},\cdots, u_{n})\\
                 \nonumber&&+(-1)^{n}\frkT(\varphi(u_n)(\theta(u_1,\cdots, u_{n-1})) )\\
                \nonumber &&+\sum_{1\le i<j\le n}(-1)^{i}\theta(u_1,\cdots,\hat{u_i},\cdots,u_{j-1},\rho(Tu_i)(u_j),u_{j+1},\cdots, u_{n}).
               \end{eqnarray*}

\item $\Omega:\Hom(\wedge^n\g,\h)\oplus \Hom(\wedge^{n-1}\g\otimes V,W)\lon \Hom(\otimes^{n}V,\h)$ is defined by
\begin{eqnarray*}
 \Omega(f_\g,f_V)(u_1,\cdots,u_n)=(-1)^{n}\Big(f_\g(Tu_1,\cdots,Tu_n)-\frkT f_V (Tu_1,\cdots,Tu_{n-1}, u_n)\Big),\quad \forall u_1,\cdots,u_n\in V.
\end{eqnarray*}
\end{itemize}
\begin{thm}\label{cohomology-of-LLT}
  With the above notations,  $(\oplus _{n=0}^{+\infty}\frkC^n(\g,\rho,T,\phi,\varphi,\frkT),\huaD_R)$ is a cochain complex, i.e. $\huaD_R\circ \huaD_R=0.$
\end{thm}
\begin{proof}
 We only give a sketch of the proof and leave details to readers. Consider the semidirect product Lie-Leibniz triple $ ((\g\oplus \h,[\cdot,\cdot]_{\phi_\h}), (V\oplus W,\rho+\phi_W+\varphi),T+\frkT)$ given in Proposition \ref{pro:semiLieLeibniz}, and the associated cochain complex $(\oplus_{n=0}^{+\infty}\frkC^n(\g\oplus\h,\rho+\phi_W+\varphi, T+\frkT),\huaD)$ given in Theorem \ref{cohomology-of-average}. It is straightforward to deduce that
  $(\oplus _{n=0}^{+\infty}\frkC^n(\g,\rho,T,\phi,\varphi,\frkT),\huaD_R)$ is a subcomplex of $(\oplus_{n=0}^{+\infty}\frkC^n(\g\oplus\h,\rho+\phi_W+\varphi, T+\frkT),\huaD)$. Thus, $\huaD_R\circ \huaD_R=0.$
\end{proof}

\begin{defi}
  The cohomology of the cochain complex $(\oplus _{n=0}^{+\infty}\frkC^n(\g,\rho,T,\phi,\varphi,\frkT),\huaD_R)$ is called the {\bf cohomology  of the Lie-Leibniz triple} $((\g,\mu),(V;\rho),T)$ with coefficients in the representation $(W\stackrel{\frkT}{\lon}\h,\phi,\varphi)$. We denote its $n$-th cohomology group by $\huaH^n(\g,\rho,T,\phi,\varphi,\frkT)$.
\end{defi}

In particular, $(W\stackrel{\frkT}{\lon}\h,\phi=0,\varphi=0)$ is naturally a representation of the Lie-Leibniz triple $((\g,\mu),(V;\rho),T)$, which is called the {\bf trivial representation}, and the corresponding $n$-th cohomology group is denoted by $\huaH^n_{\rm tri}(\g,\rho,T,\frkT)$.

At the end of this subsection we study central extensions of Lie-Leibniz triples and show that central extensions are classified by the second cohomology group $\huaH^2_{\rm tri}(\g,\rho,T,\frkT)$ as applications.
\begin{defi}
Let $((\g,[\cdot,\cdot]_\g),(V;\rho),T)$ and $((\h,[\cdot,\cdot]_\h),(W,\varrho),\frkT)$ be two  Lie-Leibniz triples. An {\bf   extension} of  $((\g,[\cdot,\cdot]_\g),(V;\rho),T)$ by $((\h,[\cdot,\cdot]_\h),(W,\varrho),\frkT)$ is a short exact sequence of Lie-Leibniz triple homomorphisms:
$$\xymatrix{
  0 \ar[r] &W \ar[d]_{\frkT}\ar[r]^{i}& \hat{V}\ar[d]_{\hat{T}}\ar[r]^{p}&V\ar[d]_{T}\ar[r]&0\\
     0\ar[r] &\h \ar[r]^{ \frki} &\hat{\g}\ar[r]^{\frkp} &\g\ar[r]&0,              }$$
where $((\hat{\g},[\cdot,\cdot]_{\hat{\g}}),(\hat{V},\hat{\rho}),\hat{T})$ is a Lie-Leibniz triple.

It is called a {\bf central extension} if  $ [\alpha,\hat{x}]_{\hat{\g}}=0$, $\hat{\rho}(\hat{x})\xi=0$ and $\hat{\rho}(\alpha)\hat{u}=0$ for all   $\alpha \in \h,~\xi\in W,~\hat{x}\in\hat{\g}$ and $\hat{u}\in\hat{V}$.
\end{defi}

In the sequel, we only consider central extensions.

\begin{defi}
A {\bf section} of a central extension $((\hat{\g},[\cdot,\cdot]_{\hat{\g}}),(\hat{V},\hat{\rho}),\hat{T})$ of a Lie-Leibniz triple $((\g,[\cdot,\cdot]_\g),(V;\rho),T)$  by   $(\h ,W ,\frkT)$ consists of  linear maps $\frks:\g\longrightarrow \hat{\g}$ and $s:V\longrightarrow \hat{V}$ such that
$$p\circ s=\Id_V,\quad \frkp\circ \frks=\Id_\g.$$
\end{defi}

\emptycomment{
Let $(\frks,s)$ be a section. Define
$$
\phi_0:\g\lon\gl(\h),\quad \phi_1:\g\lon\gl(W),\quad \varphi:V\lon\Hom(\h,W)
$$
by
\begin{eqnarray*}
  \phi_0(x)(\alpha)&=&[\frks(x),\alpha]_{\hat{\g}},\quad\forall x\in\g,~\alpha\in\h,\\
   \phi_1(x)(\xi)&=&\hat{\rho}(\frks(x))(\xi) ,\quad\forall x\in\g,~\xi\in W,\\
   \varphi(u)(\alpha)&=& \hat{\rho}(\alpha)(s(u)),\quad\forall u\in V,~\alpha\in\h.
\end{eqnarray*}

\begin{pro}
  With above notations, $(\phi=(\phi_0,\phi_1),\varphi)$ is a representation of the Lie-Leibniz triple $((\g,[\cdot,\cdot]_\g),(V;\rho),T)$ on the $2$-term complex $W\stackrel{\frkT}{\lon}\h$. Moreover, this representation does not depend on the choice of sections.
\end{pro}

\begin{proof}
  First we prove that the image of $\phi$ is in $\End (W\stackrel{\frkT}{\lon}\h)$, i.e. $\phi$ satisfies $\phi_0(x)\circ\frkT=\frkT\circ \phi_1(x).$ In fact, by we have
  $$
  \phi_0(x)(\frkT \xi)=[\frks(x),\frkT \xi]_{\hat{\g}}
  $$
  \yh{It does not hold. So it seems that we can only consider central extensions.}
\end{proof}
}

Let $(\frks,s)$ be a section of a central extension  $((\hat{\g},[\cdot,\cdot]_{\hat{\g}}),(\hat{V},\hat{\rho}),\hat{T})$. We further define
$$
\omega\in\Hom(\wedge^2\g,\h),\quad \varpi\in\Hom(\g\otimes V,W),\quad \huaT\in\Hom(V,\h)
$$
by
\begin{eqnarray*}
  \omega(x,y)&=&[\frks(x),\frks(y)]_{\hat{\g}}-\frks[x,y]_\g,\quad\forall x,y\in\g,\\
  \varpi(x,u)&=&\hat{\rho}(\frks(x))s(u)-s(\rho(x)u),\quad\forall x\in\g, u\in V,\\
  \huaT(u)&=&\hat{T}(s(u))-\frks(T(u)),\quad\forall  u\in V.
\end{eqnarray*}
Via the section $(\frks,s)$, $\hat{\g}\cong \g\oplus \h$ and $\hat{V}\cong V\oplus W$. Transfer the Lie-Leibniz triple structure on $\hat{\g}$ and $\hat{V}$ to that on $  \g\oplus \h$ and $ V\oplus W$, we obtain
\begin{eqnarray*}
  [x+\alpha,y+\beta]_{\hat{\g}}&=&[x,y]_\g +\omega(x,y),\quad\forall x,y\in\g, \alpha,\beta\in\h,\\
  \hat{\rho}(x+\alpha)(u+\xi)&=&\rho(x)u +\varpi(x,u),\quad\forall x\in\g, u\in V,\alpha\in\h,\xi\in W,\\
  \hat{T}(u+\xi)&=&Tu+\frkT\xi+\huaT u,\quad\forall  u\in V,\xi\in W.
\end{eqnarray*}

\emptycomment{
\begin{eqnarray*}
  [x+\alpha,y+\beta]_{\hat{\g}}&=&[x,y]_\g+\phi_0(x)\beta-\phi_0(y)\alpha+\omega(x,y),\quad\forall x,y\in\g, \alpha,\beta\in\h,\\
  \hat{\rho}(x+\alpha)(u+\xi)&=&\rho(x)u+\phi_1(x)\xi+\varphi(u)\alpha+\varpi(x,u),\quad\forall x\in\g, u\in V,\alpha\in\h,\xi\in W,\\
  \hat{T}(u+\xi)&=&Tu+\frkT\xi+\huaT u,\quad\forall  u\in V,\xi\in W.
\end{eqnarray*}
}

\begin{thm}\label{thm:cohomologicalclass}
  With the above notations, $(\omega,\varpi,\huaT)$ is a $2$-cocycle of the Lie-Leibniz triple $((\g,[\cdot,\cdot]_\g),(V;\rho),T)$  with  the trivial coefficients in  $W\stackrel{\frkT}{\lon}\h$. Moreover, its cohomological class does not depend on the choice of sections.
\end{thm}
\begin{proof}
 First by the fact that $[\cdot,\cdot]_{\hat{\g}}$ satisfies the Jacobi identity, we deduce that $\omega$ is $2$-cocycle of the Lie algebra $(\g,[\cdot,\cdot]_\g)$, i.e. $\dM_\CE\omega=0$. Then since $\hat{\rho}$ is a representation of the Lie algebra  $(\hat{\g},[\cdot,\cdot]_{\hat{\g}})$ on $V\oplus W$, we obtain
 \begin{eqnarray*}
 0&=&  \hat{\rho}([x,y]_{\hat{\g}})u-[ \hat{\rho}(x), \hat{\rho}(y)](u)\\
 &=&\rho([x,y]_\g)u+\varpi([x,y]_\g,u)-\rho(x)\rho(y)u-\varpi(x,\rho(y)u)+\rho(y)\rho(x)u+\varpi(y,\rho(x)u)\\
 &=&\varpi([x,y]_\g,u)+\varpi(y,\rho(x)u)-\varpi(x,\rho(y)u).
 \end{eqnarray*}
 Therefore, $\delta(\omega,\varpi)=0.$

 Finally since $\hat{T}$ is an embedding tensor, we obtain
 \begin{eqnarray*}
   0&=& [\hat{T}(u+\xi),\hat{T}(v+\eta)]_{\hat{\g}}-\hat{T}(\hat{\rho}(\hat{T}(u+\xi))(v+\eta))\\
   &=&[Tu,Tv]_\g+\omega(Tu,Tv)-T(\rho(Tu)v)-\frkT \varpi(Tu,v)-\huaT \rho(Tu)v\\
   &=&\omega(Tu,Tv)-\frkT \varpi(Tu,v)-\huaT \rho(Tu)v,
 \end{eqnarray*}
which implies that $ \Omega(\omega,\varpi)+\partial \huaT=0$. Therefore, $\huaD_R(\omega,\varpi,\huaT)=0,$ i.e. $(\omega,\varpi,\huaT)$ is a 2-cocycle.

Let $(\frks',s')$ be another section and $(\omega',\varpi',\huaT')$ be the associated  2-cocycle. Assume that $\frks'=\frks+N$ and $s'=s+S$ for $N\in\Hom(\g,\h)$ and $S\in\Hom(V,W)$. Then we have
\begin{eqnarray*}
  (\omega'-\omega)(x,y)&=&[\frks'(x),\frks'(y)]_{\hat{\g}}-\frks'[x,y]_\g-[\frks(x),\frks(y)]_{\hat{\g}}+\frks[x,y]_\g=-N([x,y]_\g)=\dM_\CE N(x,y),\\
  (\varpi'-\varpi)(x,u)&=&\hat{\rho}(\frks'(x))s'(u)-s'(\rho(x)u)-\hat{\rho}(\frks(x))s(u)+s(\rho(x)u)=-S(\rho(x)u),\\
 ( \huaT'-\huaT)u&=&\hat{T}(s'(u))-\frks'(T(u))-\hat{T}(s(u))+\frks(T(u))=\frkT Su-NTu,
\end{eqnarray*}
which implies that $(\omega',\varpi',\huaT')-(\omega,\varpi,\huaT)=\huaD_R(N,S,0)$. Thus, $(\omega',\varpi',\huaT')$ and $(\omega,\varpi,\huaT)$ are in the same cohomology class.
\end{proof}

Isomorphisms between central extensions can be obviously defined as follows.
\begin{defi}
Let $((\hat{\g},[\cdot,\cdot]_{\hat{\g}}),(\hat{V},\hat{\rho}),\hat{T})$ and $((\tilde{\g},[\cdot,\cdot]_{\tilde{\g}}),(\tilde{V},\tilde{\rho}),\tilde{T})$ be two central extensions of a Lie-Leibniz triple $((\g,[\cdot,\cdot]_\g),(V;\rho),T)$  by   $(\h ,W ,\frkT)$. They are said to be {\bf isomorphic} if there exists an isomorphism of Lie-Leibniz triples $(\kappa,\lambda)$ such that the following diagram commutes:
  \begin{equation*}
\xymatrix@!0{0\ar@{->} [rr]&& W \ar@{->} [rr] \ar'[d] [dd] \ar@{=} [rd] && \tilde{V}\ar'[d] [dd]\ar@ {->} [rr] \ar@{->} [rd]^{\lambda}&& V\ar@{=} [rd]\ar'[d] [dd]\ar@{->} [rr]&&0&\\
&0\ar@{->} [rr]&& W\ar@{->} [rr]\ar@{->} [dd]&&\hat{V}\ar@{->} [dd]\ar@{->} [rr]&&V\ar@ {->} [dd]\ar@{->} [rr]&&0\\
0\ar@{->} [rr]&&\h\ar'[r] [rr] \ar@{=} [rd]&&\tilde{\g}\ar@{->} [rd]^{\kappa}\ar'[r] [rr]&&\g\ar@{=} [rd] \ar'[r] [rr]&&0&\\
&0\ar@{->} [rr]&& \h\ar@{->} [rr]&&\hat{\g}\ar@{->} [rr]&&\g\ar@{->} [rr]&&0.}
\end{equation*}
\end{defi}

\begin{thm}
  Central extensions of a Lie-Leibniz triple $((\g,[\cdot,\cdot]_\g),(V;\rho),T)$  by   $(\h ,W ,\frkT)$ are classified by the second cohomology group $\huaH^2_{\rm tri}(\g,\rho,T,\frkT)$.
\end{thm}
\begin{proof}
 Let $((\hat{\g},[\cdot,\cdot]_{\hat{\g}}),(\hat{V},\hat{\rho}),\hat{T})$ and $((\tilde{\g},[\cdot,\cdot]_{\tilde{\g}}),(\tilde{V},\tilde{\rho}),\tilde{T})$ be two isomorphic central extensions. Let $(\tilde{\frks},\tilde{s})$ be a section of $((\tilde{\g},[\cdot,\cdot]_{\tilde{\g}}),(\tilde{V},\tilde{\rho}),\tilde{T})$, and $(\tilde{\omega},\tilde{\varpi},\tilde{\huaT})$ be the corresponding 2-cocycle. Define $(\hat{\frks},\hat{s})$ by
 $$
 \hat{\frks}=\kappa\circ\tilde{\frks},\quad \hat{s}=\lambda\circ s.
 $$
 Then it is obvious that $(\hat{\frks},\hat{s})$ is a section of $((\hat{\g},[\cdot,\cdot]_{\hat{\g}}),(\hat{V},\hat{\rho}),\hat{T})$. Let $(\hat{\omega},\hat{\varpi},\hat{\huaT})$ be the corresponding 2-cocycle. Then we have
 $$
 \hat{\omega}(x,y)=[\hat{\frks}(x),\hat{\frks}(y)]_{\hat{\g}}-\hat{\frks}[x,y]_\g=[\lambda\tilde{\frks}(x),\lambda\tilde{\frks}(y)]_{\hat{\g}}-\lambda\tilde{\frks}[x,y]_\g=\lambda([ \tilde{\frks}(x), \tilde{\frks}(y)]_{\hat{\g}}- \tilde{\frks}[x,y]_\g)=\tilde{\omega}(x,y).
 $$
 Similarly, we have $$\tilde{\varpi}=\hat{\varpi},\quad \tilde{\huaT}=\hat{\huaT}.$$
 By Theorem \ref{thm:cohomologicalclass},  isomorphic central extensions gives rise to the same cohomological class in $\huaH^2_{\rm tri}(\g,\rho,T,\frkT)$.

 The converse part can be easily checked and we omit details.
\end{proof}

\section{Homotopy embedding tensors and higher structures}\label{sec:homotopy}

In this section, we define homotopy embedding tensors and establish
various relations between homotopy embedding tensors,
Leibniz$_\infty$-algebras, $A_\infty$-algebras and
$L_\infty$-algebras. In particular, we construct $L_\infty$-algebras
from homotopy embedding tensors via Leibniz$_\infty$-algebras, which
generalizes the construction of $L_\infty$-algebras from embedding
tensors given by Kotov and Strobl in \cite{KS}. A hidden aim of us to build up homotopy
theory for embedding tensors is to try to find possible equivalence
between them to provide equivalence of the corresponding physical
models. To build up weak equivalence between homotopy embedding
tensors is our next aim which we postpone to the future. If the physical model is topological, then
weak equivalence should definitely provide a suitable equivalence. But quite possibly, weak
equivalence might not provide non-trivial equivalence between
embedding tensors themselves, just like, weak equivalences for Lie
algebras viewed as $L_\infty$-algebras are simply isomorphisms.

\subsection{Homotopy embedding tensors and  Leibniz$_\infty$-algebras}

In this subsection, we introduce the notion of a homotopy embedding tensor on an $L_\infty$-algebra and show that it induces a Leibniz$_\infty$-algebra generalizing Proposition \ref{average-Leibniz}.

Let $V^\bullet$ be a graded vector space. Denote by $\Hom^n(\bar{\Ten}(V^\bullet),V^\bullet)$ the space of degree $n$ linear maps from the graded vector space $\bar{\Ten}(V^\bullet)$ to the graded vector space $V^\bullet$. Obviously, an element $f\in\Hom^n(\bar{\Ten}(V^\bullet),V^\bullet)$ is the sum of $f_i:{{\otimes^ i}V^\bullet}\lon V^\bullet$. We will write  $f=\sum_{i=1}^{+\infty} f_i$.
 Set $\CV^n(V^\bullet,V^\bullet):=\Hom^n(\bar{\Ten}(V^\bullet),V^\bullet)$ and
$
\CV^\bullet(V^\bullet,V^\bullet):=\oplus_{n\in\mathbb Z}\CV^n(V^\bullet,V^\bullet).
$
As the graded version of the Balavoine bracket given in \cite{Bal}, the {\bf graded Balavoine bracket} $[\cdot,\cdot]_{\B}$ on the graded vector space $\CV^\bullet(V^\bullet,V^\bullet)$ is given
by:
\begin{eqnarray}
[f,g]_{\B}:=f\bar{\circ} g-(-1)^{mn}g\bar{\circ} f,\,\,\,\,\forall f=\sum_{i=1}^{+\infty} f_i\in \CV^m(V^\bullet,V^\bullet),~g=\sum_{j=1}^{+\infty}g_j\in \CV^n(V^\bullet,V^\bullet),
\label{eq:gfgcirc-B}
\end{eqnarray}
where $f\bar{\circ} g\in \CV^{m+n}(V^\bullet,V^\bullet)$ is defined by
 \begin{eqnarray}\label{graded-NR-circ-B}
f\bar{\circ} g&=&\Big(\sum_{i=1}^{+\infty}f_i\Big)\bar{\circ}\Big(\sum_{j=1}^{+\infty}g_j\Big):=\sum_{s=1}^{+\infty}\Big(\sum_{i+j=s+1}f_i\bar{\circ} g_j\Big),
\end{eqnarray}
while $f_i\bar{\circ} g_j\in \Hom({{\otimes ^s}V^\bullet},V^\bullet)$ is defined by
$f_i\bar{\circ} g_j=\sum_{k=1}^{i}f_i\bar{\circ}_k g_j
$
and $f_i\bar{\circ}_k g_j$ is defined by
\begin{equation*}
\begin{split}
&(f_i\bar{\circ}_k g_j)(v_1,\cdots,v_{s})\\
={}& \sum_{\sigma\in\mathbb S_{(k-1,j-1)}}(-1)^{\beta_k}\varepsilon(\sigma)f_i(v_{\sigma(1)},\cdots,v_{\sigma(k-1)},g_j(v_{\sigma(k)},\cdots,v_{\sigma(k+j-2)},v_{k+j-1}),v_{k+j},\cdots,v_{s}),
\end{split}
\end{equation*}
where $\beta_k=n(v_{\sigma(1)}+v_{\sigma(2)}+\cdots+v_{\sigma(k-1)})$.

Similar as the classical case, $(\CV^\bullet(V^\bullet,V^\bullet),[\cdot,\cdot]_{\B})$ is a graded Lie algebra.

The notion of a Leibniz$_\infty$-algebra was introduced in \cite{livernet}, and further studied in  \cite{ammardefiLeibnizalgebra,SL,Uchino-1}.

 \begin{defi}
A {\bf  Leibniz$_\infty$-algebra} is a $\mathbb Z$-graded vector space $\la^\bullet=\oplus_{k\in\mathbb Z}\la^k$ equipped with a collection $(k\ge 1)$ of linear maps $\oprn_k:\otimes^k\la^\bullet\lon\la^\bullet$ of degree $1$ such that  $\sum_{k=1}^{+\infty}\oprn_k$ is a Maurer-Cartan element of the graded Lie algebra $(\CV^\bullet(\la^\bullet,\la^\bullet),[\cdot,\cdot]_{\B})$. More precisely, for any homogeneous elements $x_1,\cdots,x_n\in \la^\bullet$, the following equality holds:
\begin{eqnarray*}
\sum_{i=1}^{n}\sum_{k=1}^{n-i+1}\sum_{\sigma\in \mathbb S_{(k-1,i-1)} }(-1)^{\gamma_{k}}\varepsilon(\sigma)\oprn_{n-i+1}(x_{\sigma(1)},\cdots,x_{\sigma(k-1)},\oprn_i(x_{\sigma(k)},\cdots,x_{\sigma(k+i-2)},x_{k+i-1}),x_{k+i},\cdots,x_{n})=0,
\end{eqnarray*}
where $\gamma_{k}=x_{\sigma(1)}+\cdots+x_{\sigma(k-1)}$.
\end{defi}
It is obvious that   an $L_\infty$-algebra  is naturally a Leibniz$_\infty$-algebra.

\begin{defi}
Let $(\la^\bullet,\{\oprn_k\}_{k=1}^{+\infty})$ and $({\la'}^\bullet,\{\oprn_k'\}_{k=1}^{+\infty})$ be two Leibniz$_\infty$-algebras.
A Leibniz$_\infty$-algebra {\bf homomorphism} from $(\la^\bullet,\{\oprn_k\}_{k=1}^{+\infty})$ to $({\la'}^\bullet,\{\oprn_k'\}_{k=1}^{+\infty})$ consists of a collection of degree $0$ graded multilinear maps $f_k:{\otimes ^k}\la^\bullet\lon {\la'}^\bullet,~ k\ge 1$ with the property that,
for any $n\geq 1$ and homogeneous elements $x_1,\cdots,x_n\in \la^\bullet$,
the following equality holds:
\begin{eqnarray*}
&&\sum_{i=1}^{n}\sum_{k=1}^{n-i+1}\sum_{\sigma\in \mathbb S_{(k-1,i-1)} }(-1)^{\gamma_{k}}\varepsilon(\sigma)f_{n-i+1}(x_{\sigma(1)},\cdots,x_{\sigma(k-1)},\oprn_i(x_{\sigma(k)},\cdots,x_{\sigma(k+i-2)},x_{k+i-1}),x_{k+i},\cdots,x_{n})\\
&&=\sum_{p=1}^n\sum_{\sigma\in  \mathbb E_{(k_1,\cdots,k_p)} \atop k_1+\cdots+k_p=n}\varepsilon(\sigma)\oprn_p'(f_{k_1}(x_{\sigma(1)},\cdots,x_{\sigma(k_1)}),\cdots,f_{k_p}(x_{\sigma(k_1+\cdots+k_{p-1}+1)},\cdots,x_{\sigma(n)})),
\end{eqnarray*}
where $\gamma_{k}=x_{\sigma(1)}+\cdots+x_{\sigma(k-1)}$ and $\mathbb E_{(k_1,\cdots,k_p)}$ denotes the set of shuffles $\sigma\in\mathbb S_{(k_1,\cdots,k_p)}$ such that $\sigma(k_1)<\sigma(k_1+k_2)<\cdots<\sigma(k_1+k_2+\cdots+k_p)$.
\end{defi}

\begin{propdef}Let $\g^\bullet$ and $V^\bullet$ be graded vector spaces.
Let $l_k:\Sym^{k}(\g^\bullet)\lon\g^\bullet$ and $\rho_k:\Sym^{k-1}(\g^\bullet)\otimes V^\bullet\lon V^\bullet,~ k\ge 1$ be linear maps of degree $1$. We define $l_k\boxplus\rho_k:\otimes^k(\g^\bullet\oplus V^\bullet)\lon \g^\bullet\oplus V^\bullet$ as follows
\begin{eqnarray}\label{hemisemidirect}
(l_k\boxplus\rho_k)\big((x_1,v_1),\cdots,(x_k,v_k)\big)=\big(l_k(x_1,\cdots,x_k),\rho_k(x_1,\cdots,x_{k-1},v_k)\big).
\end{eqnarray}
Then $(\g^\bullet\oplus V^\bullet,\{l_k\boxplus\rho_k\}_{k=1}^{+\infty})$ is a Leibniz$_\infty$-algebra if and only if $(\g^\bullet ,\{l_k\}_{k=1}^{+\infty})$ is an $L_\infty$-algebra and $(V^\bullet,\{\rho_k\}_{k=1}^{+\infty})$ is its  representation. This Leibniz$_\infty$-algebra is called the {\bf hemisemidirect product} of $(\g^\bullet,\{l_k\}_{k=1}^{+\infty})$ and $(V^\bullet,\{\rho_k\}_{k=1}^{+\infty})$.
\end{propdef}

A representation of an $L_\infty$-algebra will give rise to V-data.

\begin{pro}
Let $(\g^\bullet,\{l_k\}_{k=1}^{+\infty})$ be an $L_\infty$-algebra and $(V^\bullet,\{\rho_k\}_{k=1}^{+\infty})$ a representation of $(\g^\bullet,\{l_k\}_{k=1}^{+\infty})$. Then the following quadruple form V-data:
\begin{itemize}
\item[$\bullet$] the graded Lie algebra $(L,[\cdot,\cdot])$ is given by $(\CV^\bullet(\g^\bullet\oplus V^\bullet,\g^\bullet\oplus V^\bullet),[\cdot,\cdot]_\B)$;
\item[$\bullet$] the abelian graded Lie subalgebra $\h$ is given by $\h:=\oplus_{n\in\mathbb Z}\Hom^n(\bar{\Ten}(V^\bullet),\g^\bullet);$
\item[$\bullet$] $P:L\lon L$ is the projection onto the subspace $\h$;
\item[$\bullet$] $\Delta=\sum_{k=1}^{+\infty}(l_k\boxplus\rho_k)$.
\end{itemize}

Consequently, $(\h,\{\frkl_k\}_{k=1}^{+\infty})$ is an $L_\infty$-algebra, where $\frkl_k$ is given by
\eqref{V-shla}.
\end{pro}

\begin{proof}
 It is obvious that $\Img P=\h$ is an abelian graded Lie subalgebra of the g.l.a. $(\CV^\bullet(\g^\bullet\oplus V^\bullet,\g^\bullet\oplus V^\bullet),[\cdot,\cdot]_\B)$. Moreover,    $\ker P$ is also a  graded Lie subalgebra. Since $\Delta=\sum_{k=1}^{+\infty}(l_k\boxplus\rho_k)$ is the hemisemidirect product  Leibniz$_\infty$-algebra structure on $\g^\bullet\oplus V^\bullet$, we have $[\Delta,\Delta]_\B=0$ and $P(\Delta)=0$. Thus $(L,\h,P,\Delta)$ are V-data. Hence by Theorem \ref{thm:db-big-homotopy-lie-algebra}, we obtain the higher derived brackets $\{{\frkl_k}\}_{k=1}^{+\infty}$ on the abelian graded Lie subalgebra $\h$.
\end{proof}

Now we are ready to define a homotopy embedding tensor, which is the  main object   in this section. A homotopy embedding tensor on an  $L_\infty$-algebra is a generalization of an  embedding tensor on a Lie algebra.

\begin{defi}\label{homotopy-embedding-tensors-homotopy-lie}
With above notations, a degree $0$ element  $\Theta=\sum_{k=1}^{+\infty}\Theta_k\in \Hom(\bar{\Ten}(V^\bullet),\g^\bullet)$ is called a {\bf homotopy embedding tensor} on $(\g^\bullet,\{l_k\}_{k=1}^{+\infty})$ with respect to the representation $(V^\bullet,\{\rho_k\}_{k=1}^{+\infty})$ if   $\Theta=\sum_{k=1}^{+\infty}\Theta_k$ is a Maurer-Cartan element of the $L_\infty$-algebra $(\h,\{{\frkl_k}\}_{k=1}^{+\infty})$\footnote{It is a filtered  $L_\infty$-algebra \cite{Dolgushev-Rogers}. The condition of being filtered ensures convergence of the series figuring in the definition of Maurer-Cartan elements.}, that is,
\begin{eqnarray}\label{homotopy-embedding-tensors}
P\Big(e^{[\cdot,\Theta]_\B}\sum_{k=1}^{+\infty}(l_k\boxplus\rho_k)\Big)=0.
\end{eqnarray}
\end{defi}

  \begin{defi}\label{defi:isoO}
  Let $\Theta=\sum_{k=1}^{+\infty}\Theta_k\in \Hom(\bar{\Ten}(V^\bullet),\g^\bullet)$ and $\Theta'=\sum_{k=1}^{+\infty}\Theta_k'\in \Hom(\bar{\Ten}(V^\bullet),\g^\bullet)$ be homotopy embedding tensors on  $(\g^\bullet,\{l_k\}_{k=1}^{+\infty})$ with respect to the representation $(V^\bullet,\{\rho_k\}_{k=1}^{+\infty})$. A {\bf strict  homomorphism} from $\Theta'$ to $\Theta$ consists of an $L_\infty$-algebra strict homomorphism  $\phi_{\g^\bullet}:\g^\bullet\longrightarrow\g^\bullet$ and a graded linear map $\phi_{V^\bullet}:V^\bullet\longrightarrow V^\bullet$ of degree $0$ such that for any $n\geq 1$ and homogeneous elements $x_1,\cdots,x_{n-1}\in \g^\bullet,~v_1,\cdots,v_n,u\in V^\bullet$,
the following equalities hold:
      \begin{eqnarray}
       \phi_{\g^\bullet}\big(\Theta'_n(v_1,\cdots,v_n)\big)&=& \Theta_n\big(\phi_{V^\bullet}(v_1),\cdots,\phi_{V^\bullet}(v_n)\big),\label{defi:isocon11}\\
        \phi_{V^\bullet}\big(\rho_n(x_1,\cdots,x_{n-1},u)\big)&=&\rho_n\big(\phi_{\g^\bullet}(x_1),\cdots,\phi_{\g^\bullet}(x_{n-1}),\phi_{V^\bullet}(u)\big).\label{defi:isocon21}
      \end{eqnarray}
    \end{defi}

An embedding tensor induces a Leibniz algebra (Proposition \ref{average-Leibniz}). Similarly, we have the following result.

\begin{thm}\label{twist-homotopy-lie}
Let $\Theta=\sum_{k=1}^{+\infty}\Theta_k\in \Hom(\bar{\Ten}(V^\bullet),\g^\bullet)$ be a homotopy embedding tensor on $(\g^\bullet,\{l_k\}_{k=1}^{+\infty})$ with respect to the representation $(V^\bullet,\{\rho_k\}_{k=1}^{+\infty})$.
\begin{itemize}
  \item[\rm(i)] $e^{[\cdot,\Theta]_\B}\sum_{k=1}^{+\infty}(l_k\boxplus\rho_k)$ is a Maurer-Cartan element of the g.l.a. $(\CV^\bullet(\g^\bullet\oplus V^\bullet,\g^\bullet\oplus V^\bullet),[\cdot,\cdot]_\B)$;
  \item[\rm(ii)] there is a Leibniz$_\infty$-algebra structure on $V^\bullet$  given by
\begin{eqnarray}\label{double-homotopy-lie}
&&\theta_k(v_1,\cdots,v_{k})=\Big(e^{[\cdot,\Theta]_\B}\sum_{k=1}^{+\infty}(l_k\boxplus\rho_k)\Big)(v_1,\cdots,v_{k}).
\end{eqnarray}
\item[\rm(iii)] the association in (ii) gives rise to a functor $S$
  from the category of homotopy embedding tensors to that of
  Leibniz$_\infty$-algebras.
\end{itemize}
\end{thm}

\begin{proof}
(i) Since, $[\cdot,\Theta]_\B$ is a locally nilpotent derivation of $(\CV^\bullet(\g^\bullet\oplus V^\bullet,\g^\bullet\oplus V^\bullet),[\cdot,\cdot]_\B)$, we deduce that $e^{[\cdot,\Theta]_\B}$ is an automorphism of $(\CV^\bullet(\g^\bullet\oplus V^\bullet,\g^\bullet\oplus V^\bullet),[\cdot,\cdot]_\B)$. Moreover, we have
\begin{eqnarray*}
[e^{[\cdot,\Theta]_\B}\Big(\sum_{k=1}^{+\infty}(l_k\boxplus\rho_k)\Big),e^{[\cdot,\Theta]_\B}\Big(\sum_{k=1}^{+\infty}(l_k\boxplus\rho_k)\Big)]_\B=e^{[\cdot,\Theta]_\B}[\sum_{k=1}^{+\infty}(l_k\boxplus\rho_k),\sum_{k=1}^{+\infty}(l_k\boxplus\rho_k)]_\B=0,
\end{eqnarray*}
which implies that $e^{[\cdot,\Theta]_\B}\Big(\sum_{k=1}^{+\infty}(l_k\boxplus\rho_k)\Big)$ is a Maurer-Cartan element of the graded Lie algebra $(\CV^\bullet(\g^\bullet\oplus V^\bullet,\g^\bullet\oplus V^\bullet),[\cdot,\cdot]_\B)$.

(ii)  By \eqref{homotopy-embedding-tensors},   $e^{[\cdot,\Theta]_\B}\sum_{k=1}^{+\infty}(l_k\boxplus\rho_k)|_{V^\bullet}$ is a
Leibniz$_\infty$-algebra structure on $V^\bullet$.

(iii) Let $\Theta\in \Hom(\bar{\Ten}(V^\bullet),\g^\bullet)$ and
$\Theta'\in \Hom(\bar{\Ten}(V^\bullet),\g^\bullet)$ be homotopy
embedding tensors and
$(\phi_{\g^\bullet},\phi_{V^\bullet})$ a strict homomorphism from
$\Theta'$ to $\Theta$.  For any $n\geq 1$ and homogeneous elements $v_1,\cdots,v_n\in V^\bullet$, we
have
\begin{eqnarray*}
\phi_{V^\bullet}\big(\oprn_n'(v_1,\cdots,v_n)\big)&=&\phi_{V^\bullet}\Big(e^{[\cdot,\Theta']_\B}\sum_{k=1}^{+\infty}(l_k\boxplus\rho_k)\Big)(v_1,\cdots,v_{k})\\
                                              &\stackrel{\eqref{defi:isocon11},\eqref{defi:isocon21}}{=}&\Big(e^{[\cdot,\Theta]_\B}\sum_{k=1}^{+\infty}(l_k\boxplus\rho_k)\Big)(\phi_{V^\bullet}(v_1),\cdots,\phi_{V^\bullet}(v_n))\\
                                              &=&\oprn_n\big(\phi_{V^\bullet}(v_1),\cdots,\phi_{V^\bullet}(v_n)\big).
\end{eqnarray*}
Therefore,  $\phi_{V^\bullet}$ is a strict
homomorphism from the Leibniz$_\infty$-algebra
$(V^{\bullet},\{\oprn_k'\}_{k=1}^{+\infty})$ to
$(V^{\bullet},\{\oprn_k\}_{k=1}^{+\infty})$. Then it is straightforward to see that it is actually a functor.
\end{proof}

\subsection{Leibniz$_\infty$-algebras and $A_\infty$-algebras}

In this subsection, first we recall the B\"{o}rjeson products on graded associative algebras which is a useful tool to construct $A_\infty$-algebras. Then by the bar construction, we show that a Leibniz$_\infty$-algebra $\la^\bullet$ gives rise to an  $A_\infty$-algebra $\bar{\Ten}(\la^\bullet).$
\begin{defi}{\rm (\cite{Sta63})}
An {\bf $A_{\infty}$-algebra} is a $\mathbb Z$-graded vector space $ A^\bullet=\oplus_{k\in\mathbb Z} A^k$ endowed with a family of graded maps
$
m_i:\otimes^i A^\bullet\lon A^\bullet,\,deg(m_i)=1,\, i\ge 1
$
satisfying the {\bf Stasheff identities}
$$
\sum_{i=1}^{n}\sum_{k=1}^{n-i+1}(-1)^{a_1+\cdots+a_{k-1}}m_{n-i+1}(a_1,\cdots,a_{k-1}, m_i(a_{k},\cdots,a_{k+i-1}),a_{k+i},\cdots,a_n)=0,
$$
for $n\ge 1$ and any homogeneous elements $a_1,\cdots,a_n\in A^\bullet$.
\end{defi}

We recall the definition of B\"{o}rjeson products on graded associative algebras, which give rise to $A_\infty$-algebras.

\begin{defi}{\rm (\cite{Borjeson,DSV,Markl})}
Let $A^\bullet$ be a graded associative algebra, and let $\nabla:A^\bullet\lon A^\bullet$ be a degree $1$ linear map such that $\nabla\circ\nabla=0$. The sequence of {\bf B\"{o}rjeson products} $b_k^{\nabla}:{\otimes ^ n}A^\bullet\lon A^\bullet$ are defined as follows:
\begin{eqnarray*}
b_1^{\nabla}(a_1)&=&\nabla(a_1),\\
b_2^{\nabla}(a_1,a_2)&=&\nabla(a_1a_2)-\nabla(a_1)a_2-(-1)^{a_1}a_1\nabla(a_2),\\
b_3^{\nabla}(a_1,a_2,a_3)&=&\nabla(a_1a_2a_3)-\nabla(a_1a_2)a_3-(-1)^{a_1}a_1\nabla(a_2a_3)+(-1)^{a_1}a_1\nabla(a_2)a_3,\\
&\vdots&\\
b_k^{\nabla}(a_1,\cdots,a_k)&=&\nabla(a_1\cdots a_k)-\nabla(a_1\cdots a_{k-1})a_k-(-1)^{a_1}a_1\nabla(a_2\cdots a_{k})+(-1)^{a_1}a_1\nabla(a_2\cdots a_{k-1})a_k,
\end{eqnarray*}
for any homogeneous elements $a_1,\cdots,a_k\in A^\bullet$.
\end{defi}

\begin{thm}{\rm (\cite{Borjeson,DSV,Markl})}\label{thm:BA}
With the  above notations, $(A^\bullet,\{b_k^{\nabla}\}_{k=1}^{+\infty})$ is an $A_\infty$-algebra.
\end{thm}

Let $(\la^\bullet,\{\theta_k\}_{k=1}^{+\infty})$ be a
Leibniz$_\infty$-algebra. By the bar construction
\cite{ammardefiLeibnizalgebra} of a Leibniz$_\infty$-algebra (see also
\cite[Sect.11.4.3, Sect.13.5]{LV}),  we have a codifferential cofree conilpotent
coZinbiel coalgebra $
(\bar{\Ten}(\la^\bullet),\triangle,d=\sum_{k=1}^{+\infty}d_k)
$
 as following:
\begin{eqnarray*}
\triangle(x_1\otimes\cdots\otimes x_n)=\begin{cases}
0,&n=1,\\
\sum_{i=1}^{n-1}\sum_{\sigma\in \mathbb S_{(i,n-i-1)}}\varepsilon(\sigma)(x_{\sigma(1)}\otimes\cdots\otimes x_{\sigma(i)})\otimes(x_{\sigma(i+1)}\otimes\cdots\otimes x_{\sigma(n-1)}\otimes x_n) ,&n\geq2,
\end{cases}
\end{eqnarray*}
for $n<k$,~$d_k(x_1\otimes\cdots\otimes x_n)=0$ and for $n\geq k$

\begin{eqnarray*}
d_k(x_1\otimes\cdots\otimes x_n)&=&
\sum_{j=1}^{n+1-k}\sum_{\sigma\in \mathbb S_{(j-1,k-1)} }(-1)^{x_{\sigma(1)+\cdots+x_{\sigma(j-1)}}}\varepsilon(\sigma;x_1,\cdots, x_{j+k-2})\\
&&x_{\sigma(1)}\otimes\cdots \otimes x_{\sigma(j-1)}\otimes \theta_k(x_{\sigma(j)},\cdots,x_{\sigma(j+k-2)},x_{j+k-1})\otimes x_{j+k}\otimes\cdots\otimes x_n.
\end{eqnarray*}

We recall that the graded vector space $\bar{\Ten}(\la^\bullet)$ was equipped with the tensor  product $\otimes:\bar{\Ten}(\la^\bullet)\otimes \bar{\Ten}(\la^\bullet)\lon \bar{\Ten}(\la^\bullet)$ given, for $ x_1,\cdots,x_{m+n}\in\g$, by
\begin{eqnarray*}
(x_1\otimes \cdots\otimes x_n)\otimes(x_{n+1}\otimes \cdots\otimes x_{m+n})=x_1\otimes \cdots\otimes x_{m+n}.
\end{eqnarray*}
Moreover, $(\bar{\Ten}(\la^\bullet),\otimes)$ is a free graded nonunital associative algebra. Thus,  $(\bar{\Ten}(\la^\bullet),\otimes,d)$ is a graded associative algebra with a linear map $d$ such that $d\circ d=0$. By Theorem \ref{thm:BA}, we have

\begin{thm}
Let $(\la^\bullet,\{\theta_k\}_{k=1}^{+\infty})$ be a Leibniz$_\infty$-algebra. Then  $(\bar{\Ten}(\la^\bullet),\{b_k^{d}\}_{k=1}^{+\infty})$ is an $A_\infty$-algebra.
\end{thm}

\begin{rmk}
  If the Leibniz$_\infty$-algebra $(\la^\bullet,\{\theta_k\}_{k=1}^{+\infty})$  reduces to a Leibniz algebra $(\la,[\cdot,\cdot]_\la)$, there is an $A_\infty$-algebra structure on $\bar{\Ten}(\la)$.  More precisely, the linear maps $m_k$ are given by
\begin{eqnarray*}
m_1(x_1\otimes\cdots\otimes x_n)=\begin{cases}
0,&n=1,\\
\sum_{1\le i<j\le n}(-1)^{i}x_1\otimes\cdots \otimes x_{i-1}\otimes x_{i+1}\otimes\cdots\otimes x_{j-1}\otimes [x_i,x_j]_\g\otimes\cdots\otimes x_n,&n\geq2,
\end{cases}
\end{eqnarray*}
\emptycomment{
\begin{eqnarray*}
&&m_2(x_1\otimes\cdots\otimes x_n,x_{n+1}\otimes \cdots\otimes x_{n+m})\\
&=&\sum_{1\le i\le n\atop 1\le j\le m}(-1)^{i}x_1\otimes\cdots \otimes x_{i-1}\otimes x_{i+1}\otimes\cdots\otimes x_{n}\otimes x_{n+1}\otimes\cdots\otimes x_{n+j-1}\otimes [x_i,x_{n+j}]_\g\otimes x_{n+j+1}\otimes \cdots \otimes x_{n+m},
\end{eqnarray*}
}
and for $k\geq 2$

\begin{eqnarray*}
&&m_k(x_1\otimes\cdots\otimes x_{n_1},x_{n_1+1}\otimes\cdots\otimes x_{n_1+n_2},\cdots,x_{n_1+\cdots+n_{k-1}+1}\otimes\cdots\otimes x_{n_1+\cdots+n_{k-1}+n_k})\\
&=&\sum_{1\le i\le n_1\atop 1\le j\le n_k}(-1)^{i}x_1\otimes\cdots \otimes x_{i-1}\otimes x_{i+1}\otimes\cdots\otimes x_{n_1}\otimes x_{n_1+1}\otimes\cdots\otimes x_{n_1+\cdots+n_{k-1}+j-1}\otimes \\
&&[x_i,x_{n_1+\cdots+n_{k-1}+j}]_\g\otimes x_{n_1+\cdots+n_{k-1}+j+1}\otimes \cdots \otimes x_{n_1+\cdots+n_{k-1}+n_k}.
\end{eqnarray*}

\end{rmk}

\subsection{Leibniz$_\infty$-algebras and $L_\infty$-algebras}
There is a procedure to associate an
$L_\infty$-algebra to a Leibniz algebra \cite{KS}.  In this section, we extend this
construction to a functor from the category of
Leibniz$_\infty$-algebras to that of $L_\infty$-algebras. Thus
we arrive at a functor from the category of homotopy embedding tensors
to that of $L_\infty$-algebras.

Let $V^\bullet$ be a $\mathbb Z$-graded vector space. Then the tensor
algebra $(\Ten(V^\bullet),\otimes)$ is a free graded unital associative algebra.
The freeness implies the uniqueness of the graded unital algebra morphism $\triangle^{\co}:\Ten(V^\bullet)\lon\Ten(V^\bullet)\otimes \Ten(V^\bullet)$ such that
$$
\triangle^{\co}(x)=1\otimes x+x\otimes 1,\quad\forall x\in V^\bullet.
$$
More precisely, $\triangle^{\co}$ is given by
\begin{eqnarray*}
\bigtriangleup^{\co}(x_1\otimes\cdots\otimes x_n)=\sum_{i=0}^n\sum_{\sigma\in\mathbb S_{(i,n-i)}}\varepsilon(\sigma;x_1,\cdots, x_n)\big(x_{\sigma(1)} \otimes\cdots\otimes x_{\sigma(i)}\big)\otimes \big(x_{\sigma(i+1)} \otimes\cdots\otimes x_{\sigma(n)}\big).
\end{eqnarray*}
It is immediate to check that it is coassociative
and counital. Hence $(\Ten(V^\bullet),\otimes,\triangle^{\co})$ is a
graded bialgebra. We call $(\Ten(V^\bullet),\otimes,\triangle^{\co})$
the (graded) coshuffle bialgebra.

Let $\Lie(V^\bullet)$ be the free graded Lie algebra generated by the graded
vector space $V^\bullet$. In fact, $\Lie(V^\bullet)$ is  the
intersection of all the Lie subalgebras of the (graded) commutator Lie algebra
$\Ten(V^\bullet)_\Lie=(\Ten(V^\bullet), [\cdot, \cdot]_C)$ containing $V^\bullet$. Note that the space of
primitive elements of  $(\Ten(V^\bullet),\otimes,\triangle^{\co})$  is
$\Lie(V^\bullet)$. Thus the grading on $\Ten(V^\bullet)$ induces a natural grading  on
$\Lie(V^\bullet)$ making it a graded Lie algebra. We have an embedding $V^\bullet\subset\Lie(V^\bullet)$ of
graded vector spaces. This induces $\frki:\Ten(V^\bullet)\lon
\Ten(\Lie(V^\bullet))$ a grading preserving inclusion of graded
coshuffle bialgebras. Recall that the universal enveloping algebra
$\U(\Lie(V^\bullet)) = \Ten(\Lie(V^\bullet))/I_\U$, where the two sided ideal $I_\U$ is generated
by $a\otimes b - b\otimes a-[a, b]$. Moreover, we have
$$\triangle^{\co}(I_\U)\subset I_\U\otimes \Ten(\Lie(V^\bullet))+\Ten(\Lie(V^\bullet))\otimes  I_\U.$$
Thus $I_\U$  is a homogeneous ideal of the bialgebra $(\Ten(V^\bullet),\otimes,\triangle^{\co})$.
Then the natural projection $\frkp:\Ten(\Lie(V^\bullet))\lon \U(\Lie(V^\bullet))$  is a surjective homomorphism of graded
bialgebras.

\begin{lem}\label{coiso}
With the above notations, $\Phi=\frkp\circ \frki:\Ten(V^\bullet)\lon\U(\Lie(V^\bullet))$ is an isomorphism of graded bialgebras. More precisely, the isomorphism $\Phi$ is given by
\begin{eqnarray*}
\Phi(1)=1,\quad\Phi(v_1\otimes\cdots\otimes v_n)=v_1 * \cdots * v_n,~\forall v_1,\cdots,v_n\in V^\bullet,
\end{eqnarray*} where $*$ denotes the multiplication in
$\U(\Lie(V^\bullet))$.
\end{lem}

\begin{proof}
Since $\Phi$ is a homomorphism of graded bialgebras, we only need to prove that $\Phi$ is an isomorphism. Let $A$ be a unital associative algebra, $A_\Lie$ the commutator Lie algebra of $A$ and $f:\Lie(V^\bullet)\lon A_\Lie$ a Lie algebra homomorphism. Since $\Ten(V^\bullet)$ is a free graded unital associative algebra, we have a unique associative algebra homomorphism $\bar{f}:\Ten(V^\bullet)\lon A$ such that $f=\bar{f}\circ i$, where $i:\Lie(V^\bullet)\lon \Ten(V^\bullet)_\Lie$ is the inclusion of graded Lie algebras. Thus $\Ten(V^\bullet)$ satisfies the universal property of the universal enveloping algebra of the free graded Lie algebra $\Lie(V^\bullet)$. Set $A=\U(\Lie(V^\bullet))$, by the universal property, we deduce that $\frkp\circ \frki:\Ten(V^\bullet)\lon\U(\Lie(V^\bullet))$ is an isomorphism of graded associative algebras. Thus,
$\Phi$  is an isomorphism of bialgebras.
\end{proof}

\begin{rmk}
Since $(\Ten(V^\bullet),\otimes,\triangle^{\co})$ is a conilpotent
cocommutative bialgebra, Lemma \ref{coiso} is a special case of the
Cartier-Milnor-Moore theorem \cite{Milnor} (see also \cite[Sect.1.3.2]{LV}). For generalized bialgebras and an operadical version of Cartier-Milnor-Moore theorem, please see the monograph \cite{Loday}.
\end{rmk}

By Lemma \ref{coiso}, $\Phi$ is an isomorphism from the graded coaugmented coalgebra
$(\Ten^c(V^\bullet),\triangle^{\co})$ to the graded coaugmented
coalgebra $\U^c(\Lie(V^\bullet))$. Thus $\Phi$ is also an isomorphism from graded noncounital coalgebra
$(\bar{\Ten}^c(V^\bullet),\triangle^{\co})$ to the graded noncounital
coalgebra $\bar{\U}^c(\Lie(V^\bullet))$.

\begin{thm}{\rm (Poincar\'e-Birkhoff-Witt)}\label{coiso-1}~Let $(\g,[\cdot,\cdot]_\g)$ be a graded Lie algebra. Then
the symmetrization map $\Psi:\Sym^c(\g)\lon \U^c(\g)$
\begin{eqnarray*}
\Psi(1)=1,\quad\Psi(x_1\odot\cdots\odot x_m)=\frac{1}{m!}\sum_{\sigma\in\mathbb S_{m}}\varepsilon(\sigma;x_1,\cdots,
x_m)x_{\sigma(1)} * \cdots * x_{\sigma(m)},~\forall x_1,\cdots,x_m\in \g
\end{eqnarray*}
is an isomorphism of  graded coaugmented  coalgebras.
\end{thm}

\emptycomment{
\begin{thm}{\rm (Poincar\'e-Birkhoff-Witt)}\label{coiso-1}
The symmetrization map $\Psi:\Sym^c(\Lie(V^\bullet))\lon \U^c(\Lie(V^\bullet))$
\begin{eqnarray*}
\Psi(1)=1,\quad\Psi(x_1\odot\cdots\odot x_m)=\frac{1}{m!}\sum_{\sigma\in\mathbb S_{m}}\varepsilon(\sigma;x_1,\cdots,
x_m)x_{\sigma(1)} * \cdots *x_{\sigma(m)},~\forall x_1,\cdots,x_m\in \g
\end{eqnarray*}
is an isomorphism of  graded coaugmented  coalgebras.
\end{thm}
}

By Theorem \ref{coiso-1},  $\Psi$ is an isomorphism from the graded coaugmented coalgebra
$\Sym^c(\Lie(V^\bullet))$ to the graded coaugmented
coalgebra $\U^c(\Lie(V^\bullet))$.
 Thus $\Psi$ is  an isomorphism from the graded noncounital coalgebra
$\bar{\Sym}^c(\Lie(V^\bullet))$ to the graded noncounital
coalgebra $\bar{\U}^c(\Lie(V^\bullet))$. Let $(\la^\bullet,\{\theta_k\}_{k=1}^{+\infty})$ be a Leibniz$_\infty$-algebra. The bar construction of a
Leibniz$_\infty$-algebra gives us a codifferential cofree conilpotent coZinbiel coalgebra $(\bar{\Ten}(\la^\bullet),\triangle,d=\sum_{k=1}^{+\infty}d_k)$. Moreover, we have

\begin{lem}
Let $(\la^\bullet,\{\theta_k\}_{k=1}^{+\infty})$ be a Leibniz$_\infty$-algebra. Then $(\bar{\Ten}(\la^\bullet),\triangle+\tau_{12}\circ\triangle,d=\sum_{k=1}^{+\infty}d_k)$ is a codifferential coshuffle coalgebra.
\end{lem}

\begin{proof}
By the definition of $\triangle$, we deduce that $\triangle^{\co}=\triangle+\tau_{12}\circ\triangle$. Since $d=\sum_{k=1}^{+\infty}d_k$ is a
codifferential of the cofree conilpotent coZinbiel coalgebra $(\bar{\Ten}(\la^\bullet),\triangle)$, we have
\begin{eqnarray*}
\triangle^{\co}\circ d&=&(\triangle+\tau_{12}\circ\triangle)\circ d\\
                       &=&(d\otimes\Id+\Id\otimes d)\circ\triangle+\tau_{12}\circ(d\otimes\Id+\Id\otimes d)\circ\triangle\\
                       &=&(d\otimes\Id+\Id\otimes d)\circ\triangle+(d\otimes\Id+\Id\otimes d)\circ(\tau_{12}\circ\triangle)\\
                       &=&(d\otimes\Id+\Id\otimes d)\circ\triangle^{\co}.
\end{eqnarray*}
Therefore, $(\bar{\Ten}(\la^\bullet),\triangle+\tau_{12}\circ\triangle,d=\sum_{k=1}^{+\infty}d_k)$ is a codifferential coshuffle coalgebra.
\end{proof}

\begin{thm}\label{object}
With the above notations, $\Psi^{-1}\circ\Phi\circ d\circ\Phi^{-1}\circ\Psi$ is a codifferential of $\bar{\Sym}^c(\Lie(\la^\bullet))$. It gives an $L_\infty$-algebra structure on the graded
vector space $\Lie(\la^\bullet)$.
\end{thm}

\begin{proof}
Since $\Phi$ and $\Psi$ are coalgebra isomorphisms, we transfer the codifferential $d$ on $\bar{\Ten}^c(\la^\bullet)$  to $\bar{\Sym}^c(\Lie(\la^\bullet))$. Thus, $\Psi^{-1}\circ\Phi\circ d\circ\Phi^{-1}\circ\Psi$ is a codifferential of $\bar{\Sym}^c(\Lie(\la^\bullet))$, which gives an $L_\infty$-algebra structure on the graded
vector space $\Lie(\la^\bullet)$.
\end{proof}

\begin{rmk}
  This generalizes Kotov and Strobl's construction of an $L_\infty$-algebra from a Leibniz algebra \cite{KS}. In particular, by truncation, one can obtain a Lie 2-algebra, which was applied  to the supergravity theory. See also \cite{SL16} for a direct construction of a Lie 2-algebra form a Leibniz algebra.
\end{rmk}

We denote the category  of Leibniz$_\infty$-algebras and  the category of $L_\infty$-algebras by $\Leib_\infty\mbox{-}\Alg$ and $\Lie_\infty\mbox{-}\Alg$ respectively. We show that the above construction is actually a functor.

Let $f=\{f_k\}_{k=1}^{+\infty}$ be a Leibniz$_\infty$-algebra homomorphism from $(\la^\bullet,\{\oprn_k\}_{k=1}^{+\infty})$ to $({\la'}^\bullet,\{\oprn_k'\}_{k=1}^{+\infty})$.  By the bar construction of a Leibniz$_\infty$-algebra, we have a homomorphism of the codifferential cofree conilpotent coZinbiel coalgebras
$$
F:(\bar{\Ten}(\la^\bullet),\triangle,d=\sum_{k=1}^{+\infty}d_k)\lon(\bar{\Ten}({\la'}^\bullet),\triangle,d'=\sum_{k=1}^{+\infty}d_k'),
$$
which is defined by
$$
F(x_1\otimes\cdots\otimes x_n)=\sum_{p=1}^n\sum_{\sigma\in \mathbb E_{(k_1,\cdots,k_p)} \atop k_1+\cdots+k_p=n}\varepsilon(\sigma)f_{k_1}(x_{\sigma(1)},\cdots,x_{\sigma(k_1)})\otimes\cdots\otimes f_{k_p}(x_{\sigma(k_1+\cdots+k_{p-1}+1)},\cdots,x_{\sigma(n)}).
$$

\begin{lem}
With the above natation, then $F$ is a homomorphism from the codifferential coshuffle coalgebra $(\bar{\Ten}(\la^\bullet),\triangle+\tau_{12}\circ\triangle,d=\sum_{k=1}^{+\infty}d_k)$ to $(\bar{\Ten}({\la'}^\bullet),\triangle+\tau_{12}\circ\triangle,d'=\sum_{k=1}^{+\infty}d_k')$.
\end{lem}

\begin{proof}
Since $F$ is a homomorphism of codifferential cofree conilpotent coZinbiel coalgebras, we have
\begin{eqnarray*}
\triangle^{\co}\circ F&=&(\triangle+\tau_{12}\circ\triangle)\circ F\\
                       &=&(F\otimes F)\circ\triangle+\tau_{12}\circ(F\otimes F)\circ\triangle\\
                       &=&(F\otimes F)\circ\triangle+(F\otimes F)\circ(\tau_{12}\circ\triangle)\\
                       &=&(F\otimes F)\circ\triangle^{\co}
\end{eqnarray*}
and $d'\circ F=F\circ d$. Thus,   $F$ is a homomorphism of codifferential coshuffle coalgebras.
\end{proof}

\begin{thm}\label{morphism}
With the above notations, $\Psi^{-1}\circ\Phi\circ F\circ\Phi^{-1}\circ\Psi$ is a homomorphism from the codifferential  cocommutative coalgebra $\bar{\Sym}^c(\Lie(\la^\bullet))$ to $\bar{\Sym}^c(\Lie({\la'}^\bullet))$. It gives an $L_\infty$-algebra homomorphism from $\Lie(\la^\bullet)$ to $\Lie({\la'}^\bullet)$.
\end{thm}

\begin{proof}
Since $\Phi$ and $\Psi$ are coalgebra isomorphisms, we transfer the homomorphism $F:\bar{\Ten}^c(\la^\bullet)\lon\bar{\Ten}^c({\la'}^\bullet)$  to $$
\Psi^{-1}\circ\Phi\circ F\circ\Phi^{-1}\circ\Psi:\bar{\Sym}^c(\Lie(\la^\bullet))\lon\bar{\Sym}^c(\Lie({\la'}^\bullet)).
$$
 Thus, $\Psi^{-1}\circ\Phi\circ F\circ\Phi^{-1}\circ\Psi$ gives an $L_\infty$-algebra homomorphism from $\Lie(\la^\bullet)$ to $\Lie({\la'}^\bullet)$.
\end{proof}

Now summarise the results, we generalize Kotov-Strobl's construction to a functor from
the category of homotopy embedding tensors to that of $L_\infty$-algebras.

\begin{thm}\label{functor}  Theorem \ref{object} and Theorem
  \ref{morphism} give us a functor $K:\Leib_\infty\mbox{-}\Alg\lon
  \Lie_\infty\mbox{-}\Alg$. Then together with Theorem \ref{twist-homotopy-lie},
  $\KS:=K\circ S$ is a functor  from
the category of homotopy embedding tensors to that of
$L_\infty$-algebras.
\end{thm}

 \end{document}